\documentclass[letterpaper,11pt]{article}

\usepackage[utf8]{inputenc}
\usepackage[letterpaper,hmargin=1in,vmargin=1in,footskip=36pt]{geometry}
\usepackage{mathptmx} 
\DeclareMathAlphabet\mathcal{OMS}{cmsy}{m}{n}
\SetMathAlphabet\mathcal{bold}{OMS}{cmsy}{b}{n}

\usepackage[utf8]{inputenc}
\usepackage[english]{babel}
\usepackage{amsthm}
\usepackage{enumerate}
\usepackage[fleqn]{amsmath}
\usepackage{amsfonts}
\usepackage{amssymb}
\usepackage{amsbsy}
\usepackage[svgnames]{xcolor}
\usepackage{colortbl}
\usepackage{lineno}
\usepackage{xspace}

\def\ve#1{\mathchoice{\mbox{\boldmath$\displaystyle\bf#1$}}
	{\mbox{\boldmath$\textstyle\bf#1$}}
	{\mbox{\boldmath$\scriptstyle\bf#1$}}
	{\mbox{\boldmath$\scriptscriptstyle\bf#1$}}}
\newcommand{\N}{\ensuremath{\mathbb{N}}}
\newcommand{\Z}{\ensuremath{\mathbb{Z}}}

\newcommand{\G}{\ensuremath{\mathcal{G}}}
\newcommand{\A}{\ensuremath{\mathcal{A}}}
\newcommand{\apx}{$(*)$}

\makeatletter
\let\bfseries=\undefined
\DeclareRobustCommand\bfseries
{\not@math@alphabet\bfseries\mathbf
	\boldmath\fontseries\bfdefault\selectfont}
\makeatother

\def\IP{({\rm IP})_{n,\veb,\vel,\veu,\vew}}

\def\Orthant_j{{\mathcal O}_{j}}

\newcommand\veb{{\ve b}}

\newcommand\ved{{\ve d}}
\newcommand\vece{{\ve e}}
\newcommand\veg{{\ve g}}

\newcommand\veh{{\ve h}}
\newcommand\vel{{\ve l}}

\newcommand\vep{{\ve p}}

\newcommand\ver{{\ve r}}

\newcommand\veu{{\ve u}}
\newcommand\vev{{\ve v}}
\newcommand\vew{{\ve w}}
\newcommand\vex{{\ve x}}
\newcommand\vey{{\ve y}}
\newcommand\vez{{\ve z}}
\newcommand\vezero{{\ve 0}}
\newcommand\vemd{{\ve {md}}}

\newcommand\veod{{\ve {od}}}

\newcommand\vealpha{{\boldsymbol{\alpha}}}
\newcommand\vebeta{{\boldsymbol{\beta}}}

\newcommand{\td}{\textrm{td}}

\newcommand{\OO}{{\mathcal{O}}}
\newcommand{\OFPT}{{\mathcal{O}}_{FPT}}

\usepackage[textsize=tiny,textwidth=2cm,shadow,color=green]{todonotes}

\newenvironment{psmallmatrix}{\left(\smallmatrix}{\endsmallmatrix\right)}

\newcommand\FourBlockBig[5][\relax]{\begin{pmatrix}#2& #3\\#4&#5 \end{pmatrix}\ifx#1\relax\else^{(#1)}\fi}
\newcommand\FourBlock[5][\relax]{\begin{psmallmatrix}#2& #3\\#4&#5 \end{psmallmatrix}\ifx#1\relax\else{^{(#1)}}\fi}
\newcommand\TwoBlock[3][\relax]{\begin{psmallmatrix}#2\\#3 \end{psmallmatrix}\ifx#1\relax\else{^{(#1)}}\fi}

\newtheorem{theorem}{Theorem}
\newtheorem{claim}{Claim}
\newtheorem{corollary}{Corollary}
\newtheorem{lemma}{Lemma}
\newtheorem{definition}{Definition}
\newtheorem{observation}{Observation}

\makeatletter
\newtheorem*{rep@theorem}{\rep@title}
\newcommand{\newreptheorem}[2]{%
	\newenvironment{rep#1}[1]{%
		\def\rep@title{#2 \ref{##1}}%
		\begin{rep@theorem}}%
		{\end{rep@theorem}}}
\makeatother

\newreptheorem{theorem}{Theorem}
\newreptheorem{corollary}{Corollary}
\newreptheorem{lemma}{Lemma}
\newreptheorem{proposition}{Proposition}

\title{New Bounds on Augmenting Steps \\of Block-structured Integer Programs}

\author{Lin Chen\thanks{Department of Computer Science, Texas Tech University.
		\texttt{chenlin198662@gmail.com}.}
	\and Martin Kouteck{\'y}\thanks{Computer Science Institute, Charles University, and Technion - Israel Institute of Technology \texttt{koutecky@kam.mff.cuni.cz}. Partially supported by a postdoctoral fellowship at the Technion funded by the Israel Science Foundation grant 308/18, by Charles University project UNCE/SCI/004, and by the project 17-09142S of GA CR.}
	\and Lei Xu\thanks{Department of Computer Science, University of Texas Rio Grande Valley} \and Weidong Shi\thanks{Department of Computer Science, University of Houston}
}

\date{\today}

\begin{document}
	
	\maketitle
	
	\thispagestyle{empty}

\begin{abstract}
Iterative augmentation has recently emerged as an overarching method for solving Integer Programs (IP) in variable dimension, in stark contrast with the volume and flatness techniques of IP in fixed dimension.
Here we consider \emph{4-block $n$-fold integer programs}, which are the most general class considered so far.
A 4-block $n$-fold IP has a constraint matrix which consists of $n$ copies of small matrices $A$, $B$, and $D$, and one copy of $C$, in a specific block structure.
Iterative augmentation methods rely on the so-called \emph{Graver basis} of the constraint matrix, which constitutes a set of fundamental augmenting steps.

All existing algorithms rely on bounding the $\ell_1$- or $\ell_\infty$-norm of elements of the Graver basis.
Hemmecke et al.~[Math. Prog. 2014] showed that 4-block $n$-fold IP has Graver elements of $\ell_\infty$-norm at most $\OFPT(n^{2^{s_{\scalebox{.4}{D}}}})$, leading to an algorithm with a similar runtime; here, $s_{\scalebox{.6}{D}}$ is the number of rows of matrix $D$ and $\OFPT$ hides a multiplicative factor that is only dependent on the small matrices $A,B,C,D$, 
However, it remained open whether their bounds are tight, in particular, whether they could be improved to $\OFPT(1)$, perhaps at least in some restricted cases.

We prove that the $\ell_{\infty}$-norm of the Graver elements of 4-block $n$-fold IP is upper bounded by $\OFPT(n^{s_{\scalebox{.4}{D}}})$, improving significantly over the previous bound $\OFPT(n^{2^{s_{\scalebox{.4}{D}}}})$.
We also provide a matching lower bound of $\Omega(n^{s_{\scalebox{.4}{D}}})$ which even holds for arbitrary non-zero lattice elements, ruling out augmenting algorithm relying on even more restricted notions of augmentation than the Graver basis.
We then consider a special case of 4-block $n$-fold in which $C$ is a zero matrix, called 3-block $n$-fold IP.
We show that while even there the $\ell_{\infty}$-norm of its Graver elements is $\Omega(n^{s_{\scalebox{.4}{D}}})$, there exists a different decomposition into lattice elements whose $\ell_{\infty}$-norm is bounded by $\OFPT(1)$, which allows us to provide improved upper bounds on the $\ell_{\infty}$-norm of Graver elements for 3-block $n$-fold IP. 
The key difference between the respective decompositions is that a Graver basis guarantees a \emph{sign-compatible} decomposition; this property is critical in applications because it guarantees each step of the decomposition to be feasible.

Consequently, our improved upper bounds let us establish faster algorithms for 3-block $n$-fold IP and $4$-block IP, and our lower bounds strongly hint at parameterized hardness of $4$-block and even $3$-block $n$-fold IP.
\end{abstract}
\vspace{2mm}
\hspace{3.5mm}\textbf{Keywords:} Integer Programming; Graver basis; Fixed parameter tractable

\clearpage
\setcounter{page}{1}


\section{Introduction}
A powerful mathematical tool for modeling of various optimization problems is \textsc{Integer Programming}:
\begin{equation}
	\min \{\vew \cdot \vex \, : \, \A \vex = \veb, \, \vel \leq \vex \leq \veu, \, \vex \in \Z^{N}\}, \tag{IP} \label{IP}
\end{equation}
where $\vew, \veb, \vel, \veu$ are integer vectors of the \emph{objective function}, \emph{right hand side}, and \emph{lower and upper bounds}, respectively, $\A$ is an integer \emph{constraint matrix}, and $\vex$ is a vector of \emph{variables}.
It plays a key role in theory as a component in the design of approximation and parameterized algorithms, as well as in practice, with current solvers being routinely utilized in industry and capable of handling models with thousands of variables.

In general, \textsc{Integer Programming} is NP-hard, as was shown already by Karp~\cite{karp1972reducibility}, which motivates the search for tractable special cases.
Famous polynomially solvable cases are IPs with few rows and small coefficients as shown by Papadimitriou in 1981~\cite{papadimitriou1981complexity}, and IPs with few variables as shown by Lenstra in 1983.
Arguably the most significant development in the last 20 years has been the introduction of \emph{iterative augmentation} methods which led to the development of fast algorithms for wide classes of IPs whose constraint matrix has a special block structure, and to subsequent breakthrough applications in parameterized and approximation algorithms~\cite{chen2018covering,jansen2018empowering,knop2017combinatorial}.
In fact, essentially \emph{all} tractable classes of IP in variable dimension are of this kind, except for the theory of total unimodularity from the '60s.

An iterative augmentation algorithm starts with an initial feasible solution $\vex$ and iteratively finds \emph{augmenting steps} $\veg \in \Z^N$, i.e., $\vex + \veg$ is feasible and $\vew (\vex + \veg) < \vew \vex$.
A major question is \emph{where} to obtain ``good'' augmenting steps.
The \emph{Graver basis of $\A$}, $\G(\A)$, has emerged as an excellent choice, with good guarantees on convergence to optimal solutions while still being algorithmically ``tame''.
Specifically, at the heart of iterative augmentation techniques are bounds on the $\ell_1$- and $\ell_{\infty}$-norm of elements of the Graver basis, which enable dynamic programming to be used to find Graver elements.

We stress the role of bounds on the elements of $\G(\A)$.
Historically, all tractable classes of IP were discovered by proving new norm bounds  and subsequently designing a dynamic program around them, with the former typically being much harder than the latter.
Moreover, recent runtime improvements have followed from improving existing bounds~\cite{eisenbrand2018faster,martin2018parameterized}, and the most challenging questions in the field are tightly connected to norm bounds.
Our focus here is the currently least understood class of IPs, \emph{4-block $n$-fold IP}:
\begin{equation}\label{ILP:2}
	\IP: \quad \min\{\vew\cdot \vex: H \vex=\veb, \, \vel\le \vex\le \veu,\, \vex\in \Z^{t_B + nt_A} \}, 
\end{equation}
where $H$ (called a \emph{$4$-block $n$-fold matrix}) is build from smaller blocks $A$, $B$, $C$ and $D$:
\[
H=\FourBlockBig[n]CDBA :=
\begin{pmatrix}
C & D & D & \cdots & D \\
B & A & 0  &   & 0  \\
B & 0  & A &   & 0  \\
\vdots &   &   & \ddots &   \\
B & 0  & 0  &   & A
\end{pmatrix} \enspace .
\]
Here, $A,B,C,D$ are $s_i\times t_i$ matrices, $i=A,B,C,D$, respectively, and $H$ consists of $n$ copies of $A,B,D$ and one copy of $C$. Notice that by plugging $A,B,C,D$ into the above block structure we require that $s_C=s_D$, $s_A=s_B$, $t_B=t_C$ and $t_A=t_D$. Let $\Delta$ be the largest absolute value among all the entries of $A,B,C,D$.
Let $H_0$ be a matrix obtained from $H$ by setting $C=\vezero$.
We also study \emph{3-block $n$-fold IP}, obtained by replacing $H$ with $H_0$.

$4$-block $n$-fold IP remains the simplest case of block-structured IPs for which an algorithm of runtime $f(s_A, t_A, \dots, s_D, t_D, \Delta) n^{O(1)}$ (i.e., an FPT algorithm; see below) remains unknown.
From another perspective, Koutecký et al.~\cite{martin2018parameterized} has recently resolved the complexity of IP with respect to the structural parameters \emph{primal} and \emph{dual treedepth} $\td_P$ and $\td_D$, respectively, by showing that IPs with small $\td_P$ and $\td_D$ are efficiently solveable.
IPs with small \emph{incidence treedepth} $\td_I$ subsume both of the aforementioned classes as well as $4$-block $n$-fold IP, and $4$-block $n$-fold IP remains the simplest open case with respect to $\td_I$.

\subsection{Our Contribution}
Because we are interested in efficient algorithms, we wish to confine the exponential depenendence on the input into the small numbers $s_i, t_i$, $i=A,B,C,D$, and $\Delta$.
Thus we take the perspective of \emph{parameterized complexity}: for a problem instance $I$ with a \emph{parameter} $k$, we call an algorithm with runtime $f(k) |I|^{O(1)}$ a \emph{fixed-parameter tractable (FPT)} algorithm, and an algorithm with runtime $|I|^{f(k)}$ an \emph{XP algorithm} (for \emph{slice-wise polynomial}).
If such algorithms exist, we say that the problem is \emph{FPT} or \emph{XP parameterized by $k$}, respectively.
We denote by $\OFPT(|I|^{f(k)}$ a runtime of form $f'(k_1, \dots, k_\ell) |I|^{f(k)}$, i.e., with an XP dependence on $k$ but an FPT dependence on $k_1, \dots, k_\ell$, in order to focus on the (possibly dominant) dependence on $k$.

\medskip

In this paper, we provide new and improved upper bounds and resulting algorithms for 4-block and 3-block $n$-fold IP, as well as the very first lower bounds for these classes which we believe to hint at the parameterized hardness of these problems.
We denote by $\ker_{\Z}(H) = \{\vex \in \Z^{t_B + nt_A} \mid H \vex = \vezero\}$ the \emph{integer kernel of $H$}, also called the \emph{lattice of $H$}, and by $g_{\infty}(H) = \max_{\veg \in \G(H)} \|\veg\|_\infty$ the largest $\ell_\infty$-norm of an element of the Graver basis $\G(H)$; analogously for $H_0$.
First, we show an upper bound on $g_\infty(H)$.
\begin{theorem}\label{thm:3-block-graver-4}
	For any 4-block $n$-fold matrix $H$, $g_\infty(H) \leq \OFPT(n^{s_{\scalebox{.5}{D}}})$.
\end{theorem}
This improves on the previous bound of $\OFPT(n^{2^{s_{\scalebox{.4}{D}}}})$~\cite{hemmecke2014graver}.
We also establish the first explicit lower bound matching our upper bound, making it tight up to an FPT factor.
Importantly, our lower bound even applies to the first $t_B$ coordinates (denoted $\vex^0$ for a vector $\vex\in\Z^{t_B+nt_A}$) which play a special role in algorithms for $4$-block $n$-fold IP.
What is more, our lower bound even applies to \emph{any} non-zero element of $\ker_{\Z}(H)$:
\begin{theorem}\label{thm:4-block-lower}
	There exists a $4$-block $n$-fold matrix $H$ and an integer $t \in \N$ such that $s_i, t_i \in O(t)$ for $i=A,B,C,D$, and for \emph{any} $\veg \in \ker_{\Z}(H)$ we have $\|\veg^0\|_\infty = \Omega(n^t)$.
\end{theorem}
Therefore, even augmenting via a different set of steps may have to deal with steps that are unbounded by $\OFPT(1)$.
Combining Theorem~\ref{thm:3-block-graver-4} with the original idea of Hemmecke et al.~\cite{hemmecke2014graver} and a strongly polynomial framework of Koutecký et al.~\cite{martin2018parameterized}, we obtain the currently fastest algorithm for $4$-block $n$-fold IP:
\begin{theorem}\label{thm:alg-4-block}
	$4$-block $n$-fold IP can be solved in time $\OFPT(n^{O(s_{\scalebox{.5}{D}}t_B)})$.
\end{theorem}

Second, we restrict our attention to 3-block $n$-fold IP.
Unlike before, we show that its lattice elements (i.e., augmenting step candidates) admit a decomposition with $\ell_{\infty}$-norm bounded by $\OFPT(1)$:
\begin{theorem}\label{lemma:3-infty-bound}
	Any $\veg \in \ker_{\Z}(H_0)$ decomposes to $\sum_{i=1}^{N} \vece_i$ with $\vece_i \in \ker_{\Z}(H_0)$ and $\|\vece_i\|_\infty \leq \OFPT(1)$ for each $i$.
\end{theorem}
However, this decomposition is not ``sign-compatible'', meaning possibly none of its elements is a feasible step on its own, which makes its immediate algorithmic use complicated.
Nevertheless, we are able to use it to establish an upper bound of $\min\{\OFPT(n^{s_{\scalebox{.5}{D}}}),\OFPT(n^{t_A^2}+1)\}$ (below, and Theorem~\ref{thm:3-block-graver-4}):
\begin{theorem}\label{thm:3-block-graver}
	For any $3$-block $n$-fold matrix $H_0$, $g_{\infty}(H_0) \leq \OFPT(n^{t_A^2 + 1})$.
\end{theorem}
This upper bound of $\OFPT(n^{t_A^2 + 1})$, which is singly exponential in $t_A$, is much more involved compared with the upper bound of Theorem~\ref{thm:3-block-graver-4}.
This coincides with the existing results for $4$-block $n$-fold IP~\cite{hemmecke2014graver}, where an upper bound depending on $A,B$ (instead of $C,D$) is much more complicated.
Our proof relies on a completely new approach, which first establishes the decomposition of Theorem~\ref{lemma:3-infty-bound} and then modifies it into a sign-compatible decomposition through merging summands.
This may be of separate interest for deriving upper bounds on $g_\infty(\A)$ for other classes of matrices $\A$, particularly for deriving an upper bound on $g_\infty(H)$ which has an explicit dependency on $s_A,s_B,t_A,t_B$ in the exponent of $n$. 
Moreover, we show that any $4$-block $n$-fold IP can be embedded in a $3$-block $n$-fold IP (Theorem~\ref{thm:4block-3block}) in a particular way, which allows us to transfer the $4$-block $n$-fold lower bound (now restricted to \emph{feasible} lattice elements):
\begin{theorem}\label{thm:3-block-better-lower}
	There exists a $3$-block $n$-fold IP with a matrix $H$ and an integer $t \in \N$ such that $s_i, t_i \in O(t)$ for $i=A,B,C,D$, and for any \emph{feasible} nonzero $\veg \in \ker_{\Z}(H_0)$ we have $\|\veg^0\|_\infty = \Omega(n^t)$.
\end{theorem}
Finally, using our new upper bound of Theorem~\ref{thm:3-block-graver}, we get that:
\begin{theorem}\label{thm:alg-3-block}
	3-block $n$-fold IP can be solved in time $\min\{\OFPT(n^{O(s_{\scalebox{.5}{D}} t_B)}, \OFPT(n^{O(t_A^2 t_B)})\}$.
\end{theorem}


\subsection{Related Work}
$4$-block $n$-fold IP originated as a generalization of two previously studied classes of IP, the \emph{$n$-fold} and \emph{2-stage stochastic} IP, which are obtained by substituting $C=B=\vezero$ and $C=D=\vezero$, respectively, resulting in
\[
E:=
\begin{pmatrix}
D & D & \cdots & D \\
A & 0  &   & 0  \\
0  & A &   & 0  \\
\vdots &    & \ddots &   \\
0  & 0  &   & A
\end{pmatrix}\hspace{20mm}
F:=
\begin{pmatrix}
B & A & 0  &   & 0  \\
B & 0  & A &   & 0  \\
\vdots &   &   & \ddots &   \\
B & 0  & 0  &   & A
\end{pmatrix},
\]
with $E$ being the $n$-fold matrix and $F$ being the $2$-stage stochastic matrix, respectively.
The origins of iterative augmentation methods for $2$-stage stochastic IP reach the work of Hemmecke and Schultz in 2001~\cite{hemmecke2001decomposition}.
De Loera et al.~\cite{deloera2008nfold} first studied $n$-fold IP in 2008.
Later, Hemmecke et al.~\cite{hemmecke2013n} showed an FPT algorithm for $n$-fold IP based on dynamic programming, which led to a breakthrough in computational social choice~\cite{knop2017voting} and was also applied in the context of scheduling by Knop and Koutecký~\cite{knop2016scheduling}.
Later, this FPT algorithm inspired a better algorithm for a special case of \emph{combinatorial $n$-fold IP} developed by Knop et al.~\cite{knop2017combinatorial}, who also apply it to problems in stringology and graph algorithms.
Finally, this algorithm was lifted to the general $n$-fold IP by Koutecký et al.~\cite{martin2018parameterized} and Eisenbrand et al.~\cite{eisenbrand2018faster}.

An extension of $n$-fold IP to tree-structured matrices called \emph{tree-fold IP} was developed by Chen and Marx~\cite{chen2018covering} and applied to scheduling problems.
Jansen et al.~\cite{jansen2018empowering} have used $n$-fold IP to obtain efficient PTASes for scheduling problems.
An extension of $2$-stage stochastic IP analogous to tree-folds is called multi-stage stochastic and was studied by Aschenbrenner and Hemmecke~\cite{aschenbrenner2007finiteness}.
Ganian and Ordyniak~\cite{ganian2018landscape} studied the structural parameters primal treedepth and treewidth, and later Ganian et al.~\cite{ganian2017beyond} studied dual and incidence treedepth and treewidth.
Koutecký et al.~\cite{martin2018parameterized} discovered that tree-fold and multi-stage stochastic IPs are essentially equivalent to IPs with small dual and primal treedepth, settling the parameterized complexity with respect to these parameters.
The work of Koutecký et al.~\cite{martin2018parameterized} subsumes essentially all current knowledge about the solvability of IP in variable dimension with the exception of totally unimodular constraint matrices and two related classes~\cite{AppaKP:2007,ArtmannWZ:2017}, with the main remaining open problem being the complexity of $4$-block $n$-fold IP and, more generally, IP with respect to incidence treedepth.

Bounds on $g_\infty(\A)$ and $g_1(\A) = \max_{\veg \in \G(\A)} \|\veg\|_1$ play a central role in the recent developments.
For example, Chen and Marx~\cite{chen2018covering} showed that tree-fold IP is FPT, but a naive analysis yields a tower-of-exponentials dependence on the parameters.
Eisenbrand et al.~\cite{eisenbrand2018faster} lower this to double-exponential by improving the bounds on $g_1(\A)$, and, at least with the current approach, the only way to obtain a single-exponential algorithm is by obtaining single-exponential bounds on $g_1(\A)$.
Similarly, while it is known that $2$-stage stochastic IP is FPT~\cite{hemmecke2001decomposition}, there are \emph{no} known bounds at all for this algorithm except for the computability of the parameter dependence $f$ due to no bounds being available for $g_\infty(F)$.
Lower bounds on $g_\infty(\A)$ have been rare so far.
Finhold and Hemmecke~\cite{finhold2016lower} study them in the context of $n$-fold IP.
Koutecký et al.~\cite{martin2018parameterized} show lower bounds (only using elementary techniques) for IPs in terms of their primal and dual treewidth.

We use the Steinitz Lemma, which has recently gained renewed attention~\cite{eisenbrand2018proximity,eisenbrand2018faster,jansen2018integer}.

\noindent\textbf{Organization.} Statements whose proofs are deferred to Sections A--E (appendices) are marked \apx.

\section{Preliminaries}\label{sec:pre}
\noindent\textbf{Notation.}
We write vectors in boldface, e.g. $\vex, \vey$, and their entries in normal font, e.g. $x_i, y_i$.
Any $(t_B+nt_A)$-dimensional vector $\vex$ can be divided into $n+1$ \emph{bricks}, such that $\vex=(\vex^0,\vex^1,\cdots,\vex^n)$ where $\vex^0 \in \Z^{t_B}$ and each $\vex^i \in \Z^{t_A}$, $1\le i\le n$. We call $\vex^i$ the \emph{$i$-th brick} for $0\le i\le n$. We write $0_{s\times t}$ for an $s\times t$ matrix consisting of $0$, and $I_{t}$ for an $t\times t$ identity matrix. For a vector or a matrix, we write $\|\cdot\|_{\infty}$ to denote the maximal absolute value of its elements. For two vectors $\vex,\vey$ of the same dimension, $\vex\cdot\vey$ denotes their inner product.

Throughout this paper, we write $\OFPT(1)$ to represent a parameter that is only dependent on $\Delta,s_A,s_B,s_C,s_D,t_A,t_B,t_C,t_D$ where $\Delta$ is the maximal absolute value among all the entries of $A,B,C,D$, that is, $\OFPT(1)$ is only dependent on the small matrices $A,B,C,D$ and is independent of $n$. For any computable function $f(x)$, we write $\OFPT(f)$ to represent a computable function $f'(x)$ such that $|f'(x)|\le \OFPT(1)\cdot |f(x)|$, and $\Omega_{FPT}(f)$ to represent a function $f''$ such that $|f''(x)|\ge \Omega(1)\cdot |f(x)|$.

Two vectors $\vex$ and $\vey$ are called \emph{sign-compatible} if $x_i\cdot y_i\ge 0$ holds for every pair of coordinates $(x_i,y_i)$. Furthermore, we call a summation $\sum_{i}\vex_i$ \emph{sign-compatible} if the summands are pair-wise sign-compatible.

\smallskip
\noindent\textbf{Graver basis.}
Consider the general integer linear programming in the standard form~\eqref{IP}.
Let $\sqsubseteq$ be the \emph{conformal order} in $\mathbb{R}^m$ defined such that $\vex \sqsubseteq \vey$ if $\vex$ and $\vey$ lie in the same orthant, i.e., $x_i \cdot y_i \geq 0$ for each $i=1, \dots, m$, and $|x_i| \leq |y_i|$ for each $i=1, \dots, m$.
Given any subset $X\subseteq \mathbb{R}^n$, we say $\vex$ is an \emph{$\sqsubseteq$-minimal} element of $X$ if $\vex\in X$ and there does not exist $\vey\in X$, $\vey\neq \vex$ such that $\vey\sqsubseteq \vex$. It is known that every subset of $\Z^m$ has finitely many $\sqsubseteq$-minimal elements.
We study the Graver basis:
\begin{definition}[{Graver basis~\cite{graver1975foundations}}]
	The \emph{Graver basis} of an integer matrix $E$ is the
	finite set $\G(E)\subseteq \ker_{Z}(E)$ of all $\sqsubseteq$-minimal elements of $\ker_{\Z}(E) \setminus \{\vezero\}$.
\end{definition}
For clarity, we sometimes emphasize that $\veg$ comes from $\G(H)$ by writing it as $\veg(H)$, and similarly for other vectors.
We use the fact that any $\vex\in \ker_{\Z}(H)$, $\vex\neq 0$ can be written as $\vex=\sum_i\alpha_i\veg_i(H)$, where $\alpha_i\in\Z_+$, $\veg_i(H)\in \G(H)$ and $\veg_i(H)\sqsubseteq \vex$~\cite[Lemma 3.4]{onn2010nonlinear}.


\smallskip
\noindent\textbf{Augmentation algorithms for IP and Graver-best oracle.}
There is a general framework for solving~\eqref{IP} by utilizing $\G(\A)$, which was developed in a series of papers~\cite{chen2018covering,hemmecke2013n,jansen2018empowering,knop2017combinatorial}.
A recent paper by Koutecký et al.~\cite{martin2018parameterized} formalizes this framework and extends it to also obtaining \emph{strongly polynomial} algorithms (algorithms whose number of arithmetic operations does not depend on the length of the numbers on input).

We say that $\vex$ is \emph{feasible} for~\eqref{IP} if $\A\vex = \veb$ and $\vel \leq \vex \leq \veu$.
Let $\vex$ be a feasible solution for~\eqref{IP}.
We call $\veg$ a \emph{feasible step} if $\vex + \veg$ is feasible for~\eqref{IP}.
Further, call a feasible step $\veg$ \emph{augmenting} if $\vew(\vex+\veg) < \vew(\vex)$.
An augmenting step $\veg$ and a \emph{step length} $\rho \in \Z$ form an \emph{$\vex$-feasible step pair} with respect to a feasible solution $\vex$ if $\vel \le \vex + \rho\veg \le \veu$.
An augmenting step $\veh$ is a \emph{Graver-best step} for $\vex$ if $\vew(\vex + \veh) \leq \vew(\vex + \rho \veg)$ for all $\vex$-feasible step pairs $(\veg,\rho) \in \G(\A)\times \Z$.
The next definition and theorem show that it is sufficient to focus all our attention on finding Graver-best steps.
This takes care of matters such as finding an initial feasible solution, using a proximity theorem to shrink $\vew, \veb, \vel, \veu$ and so on.
\begin{definition}[Graver-best oracle] \label{def:gb_oracle}
	A {\em Graver-best oracle} for an integer matrix $\A$ is one that, queried on $\vew,\veb,\vel,\veu$
	and $\vex$ feasible to~\eqref{IP}, returns a Graver-best step $\veh$ for $\vex$.
\end{definition}
\begin{theorem}\label{thm:Koutecky}\cite{martin2018parameterized}
	Given a Graver-best oracle for $E$, (\ref{IP}) can be solved in strongly polynomial oracle time.
\end{theorem}
We note in passing that the polynomial dependence on the dimension $N$ and in particular the number of bricks $n$ when it comes to $4$-block $n$-fold IP, can be reduced using the notion of an \emph{approximate Graver-best oracle} introduced by Altmanová et al.~\cite{altmanova2018evaluating} and implicitly by Eisenbrand et al.~\cite{eisenbrand2018faster}.

\smallskip
\noindent\textbf{Finiteness theorems for $n$-fold and 2-stage stochastic matrices.}
Consider an $n$-fold matrix $E$ that consists of $A$ and $D$ (i.e., $B=C=0$ in a 4-block $n$-fold matrix).
It is shown that $g_\infty(E)$ is $\OFPT(1)$. More precisely, we have the following lemma.
\begin{lemma}\cite{hemmecke2013n,hocsten2007finiteness}\label{lemma:cite-nfold}
	Let $E$ be an $n$-fold matrix. There exists some integer $\kappa=f_{nf}(s_A,s_D,t_A,t_D,\Delta)$ for some computable function $f_{nf}$ and
	\begin{eqnarray*}
		M(E)=\{\veh\in \mathbb{Z}^{t_A} \mid \veh \text{ is the sum of at most $\kappa$ elements of } \mathcal{G}(A)\},
	\end{eqnarray*}	
	such that for any $\veg=(\veg^1,\veg^2,\cdots,\veg^n)\in\mathcal{G}(E)$
	we have $\sum_{i\in I}\veg^i\in M(A)$ for any $I\subseteq \{1,2,\cdots,n\}$.
\end{lemma}

\begin{lemma}\cite{aschenbrenner2007finiteness,martin2018parameterized}\label{lemma:multi-stage-bounded}
	Let $F$ be a two-stage stochastic matrix, and $\Delta=\max\{\|A\|_{\infty},\|B\|_{\infty}\}$. Then $g_{\infty}(H)\le f_{sto}(s_A,t_A,s_B,t_B,\Delta)$ for some computable function $f_{sto}$.
\end{lemma}
Both lemmas actually hold for even more general classes of tree-fold and multi-stage stochastic matrices.

\smallskip
\noindent\textbf{The Steinitz lemma.} The Steinitz lemma has been utilized in several recent papers~\cite{eisenbrand2018faster,eisenbrand2018proximity,jansen2018integer} to establish better algorithms for IP. We use it as well.
\begin{lemma}\cite{grinberg1980value}\label{lemma:steinitz}
	Let an arbitrary norm be given in $\mathbb{R}^{\kappa}$ and assume that $\|\vex_i\|\le \zeta$ for $1\le i\le m$ and $\sum_{i=1}^m \vex_i=\vex$. Then there exists a permutation $\pi$ such that for all positive integers $\ell\le m$, 
	$$\|\sum_{i=1}^{\ell}\vex_{\pi(i)}-\frac{\ell-\kappa}{m}\vex\|\le \kappa\zeta.$$
\end{lemma}

\section{4-block $n$-fold IP}
In this section we consider IP~(\ref{ILP:2}) for arbitrary $H$ and derive matching upper and lower bounds on the $\ell_{\infty}$-norm of its Graver basis depending on the parameter $s_{\scalebox{.6}{C}}=s_{\scalebox{.6}{D}}$.

We first establish the following upper bound that improves significantly the current result.
\begin{reptheorem}{thm:3-block-graver-4}
	For any $4$-block $n$-fold matrix $H$, $g_{\infty}(H)\le 
	\OFPT(n^{s_{\scalebox{.4}{D}}})$.	
\end{reptheorem}

\begin{proof}
	Let $\veg\in \G(H)$. As $F\cdot \veg=0$, there exist $\alpha_j\in\Z_+$, $\veg_j(F)\in\G(F)$ and $\veg_j(F)\sqsubseteq \veg$ such that 
	$$\veg=\sum_{j=1}^m \alpha_j\veg_j(F).$$ 
	Furthermore, $\|\veg_j(F)\|_{\infty}=\OFPT(1)$ according to Lemma~\ref{lemma:multi-stage-bounded}. 
	
	Let $\veh_j=C\cdot\veg^0_j(F)+\sum_{i=1}^nD\veg_j^i(F)$, which is an $s_{\scalebox{.6}{D}}$-dimensional vector such that $\|\veh_j\|_{\infty}=\OFPT(n)$. As $H\veg=0$, it follows that
	$$\sum_{j=1}^m \alpha_j\veh_j=\underbrace{\veh_1+\veh_1+\cdots+\veh_1}_{\alpha_1}+\underbrace{\veh_2+\veh_2+\cdots+\veh_2}_{\alpha_2}+\cdots+\underbrace{\veh_m+\veh_m+\cdots+\veh_m}_{\alpha_m}=0,$$
	i.e., the sequence of $\veh_i$'s sum up to $0$. According to Lemma~\ref{lemma:steinitz}, there exists a permutation of the sequence such that 
	$$\|\sum_{i=1}^{\ell}\vez_i\|_{\infty}\le s_{\scalebox{.6}{D}} \cdot \OFPT(n)=\OFPT(n), \quad \forall \ell\le m'.$$
	where $m'=\sum_{i=1}^m\alpha_i$ and $\vez_1,\vez_2,\cdots,\vez_{m'}$ is a permutation of the sequence $\underbrace{\veh_1,\veh_1,\cdots,\veh_1}_{\alpha_1},\underbrace{\veh_2,\veh_2,\cdots,\veh_2}_{\alpha_2},\cdots,\underbrace{\veh_m,\veh_m,\cdots,\veh_m}_{\alpha_m}$. Let $\tau=\OFPT(n)$ be the upper bound on $\|\sum_{i=1}^{\ell}\vez_i\|_{\infty}$, then we know that $\sum_{i=1}^{\ell}\vez_i\in \{-\tau,-\tau+1,\cdots,\tau\}^{s_{\scalebox{.4}{D}}}$. Consequently, if $m'>(2\tau+1)^{s_{\scalebox{.4}{D}}}+1$, there exists $\ell_1<\ell_2$ such that $\sum_{i=1}^{\ell_1}\vez_i=\sum_{i=1}^{\ell_2}\vez_i$, i.e., $\sum_{i=1}^{\ell_2-\ell_1}\vez_i=0$. Recall that every $\vez_i$ corresponds to some $\veh_{i'}$. Suppose $\sum_{i=1}^{\ell_2-\ell_1}\vez_i=\sum_{j=1}^m\alpha'_j\veh_j$ for $\alpha_{j}'\le \alpha_j$, then by the definition of $\veh_j$ it follows that  
	$$C\left(\sum_{j=1}^{m}\alpha_j'\veg_j^0(F)\right)+\sum_{i=1}^nD\left(\sum_{j=1}^m\alpha'_j\veg_j^i(F)\right)=0.$$
	Hence, $H\sum_{j=1}^m\alpha_j'\veg_j(F)=0$. That is, if $m'=\sum_{j=1}^m\alpha_j>(2\tau+1)^{s_{\scalebox{.4}{D}}}+1$, then there exists some $\veg'=\alpha_j'\veg_j(F)$ such that $H\veg'=0$ and $\veg'\sqsubset \veg$, contradicting the fact that $\veg\in \G(H)$. Thus, $\sum_{j=1}^m\alpha_j\le (2\tau+1)^{s_{\scalebox{.4}{D}}}+1$, implying that $\|\veg\|_{\infty}=\OFPT(n^{s_{\scalebox{.4}{D}}})$.
\end{proof}

We complement our upper bound by establishing a matching lower bound. We remark that lower bound from Theorem~\ref{thm:4-block-lower} not only holds for the $\ell_{\infty}$-norm of Graver basis elements, but even holds for any non-zero lattice element. This gives a sharp contrast to 3-block $n$-fold IP. As we will show later in Theorem~\ref{thm:3-block-better-lower} and Theorem~\ref{lemma:3-infty-bound}, a similar lower bound also exists for the $\ell_{\infty}$-norm of Graver basis elements of $\ker_{\Z}(H_0)$, however, $\ker_{\Z}(H_0)$ does admit a decomposition into lattice elements whose $\ell_{\infty}$-norm is bounded by $\OFPT(1)$.
\begin{reptheorem}{thm:4-block-lower}[\apx]
	There exists a $4$-block $n$-fold matrix $H$ such that $s_C=s_D=t-1$, $t_C=t_D=t$, $s_A=s_B=t_A=t_B=t$, and for any nonzero $\vey\in \ker_{\Z}(H)$, $\|\vey\|_{\infty}=\Omega({n^{t-1}})$, and in particular, $\|\vey^0\|_\infty = \Omega(n^{t-1})$.
\end{reptheorem}

\section{3-block $n$-fold IP}\label{sec:3block}
In this section we focus on 3-block $n$-fold IP where $H=H_0$, i.e., $C=\vezero$. As we will show in this section, 3-block $n$-fold IP admits several properties that make it a particularly interesting and important special case.
First, 3-block $n$-fold IP is without loss of generality -- any 4-block $n$-fold IP reduces to 3-block $n$-fold IP with a constant increase in the parameters.
Second, any element of $\ker_{\Z}(H_0)$ admits a decomposition into lattice elements with bounded $\ell_{\infty}$-norm, which is in a certain contrast to Theorem~\ref{thm:4-block-lower}.
Unfortunately, a strong lower bound of $\Omega(n^{t})$ for \emph{feasible} lattice elements still exists for $s_i=t_i=\OO(t)$.
Nevertheless, we establish an alternative upper bound of $\OFPT(n^{t_A^2+1})$ on the $\ell_{\infty}$-norm of the Graver basis elements for 3-block $n$-fold IP which only depends on parameters of $A$.

\subsection{Decomposition into lattice elements with bounded $\ell_{\infty}$-norm}
The goal of this subsection is to prove the following theorem.
\begin{reptheorem}{lemma:3-infty-bound}
	There exists some $\xi=\OFPT(1)$ such that for any $\veg\in\Z^{t_B+nt_A}$ satisfying $H_0\veg=0$, there exist a finite sequence of $N$ vectors $\vece_1,\vece_2,\cdots, \vece_N$ such that $\vece_h\in\Z^{t_B+nt_A}$, $H_0\vece_h=0$, $\|\vece_h\|_{\infty}\le \xi$, $\vece^0_h\sqsubseteq \veg^0$ and $\veg=\sum_h\vece_h$.
\end{reptheorem}


\begin{proof}
Since $H_0\veg=0$,
we know that $F\cdot\veg=0$. Therefore, there exist $\alpha_j\in\Z_+$, $\veg_j(F)\sqsubseteq \veg$ such that
$$\veg=\sum_{j}\alpha_j\veg_j(F),$$
where $\veg_j(F)\in\G(F)$. Consider each $\veg_j(F)$. As $F$ is a two-stage stochastic matrix, by Lemma~\ref{lemma:multi-stage-bounded} it holds for every $j$ that 
$\|\veg_j(F)\|_{\infty}=\OFPT(1).$
Note that each $\veg_j(F)$ can be written into $n+1$ bricks such that $\veg_j(F)=\left(\veg_j^0(F),\veg_j^1(F),\cdots,\veg_j^{n}(F)\right)$ where $\veg_j^0(F)$ is a $t_B$-dimensional vector, and $\veg_j^i(F)$ is a $t_A$-dimensional vector for every $1\le i\le n$. It is obvious that $\|\veg_j^i(F)\|_{\infty}=\OFPT(1)$ for every $0\le i\le n$, and it holds that
$$B\veg_j^0(F)+A\veg_j^i(F)=0, \quad \forall 1\le i\le n.$$
We first prove the following claim.

\begin{claim}\label{claim:infbound}
	For every $\veg_j(F)$ and $1\le \ell\le |\G(A)|$, there exist some $\vev_j^*$, $\alpha_{j,\ell}^i\in\Z_{\ge 0}$ such that
	\begin{itemize}
		\item $\veg_{j}^i(F)-\vev_j^*=\sum_{\ell=1}^{|\G(A)|} \alpha_{j,\ell}^i\veg_\ell(A), \quad \forall 1\le i\le n.$
		\item For every $1\le \ell\le |\G(A)|$, either $|\{i:\alpha_{j,\ell}^i>0\}|= 0$, or $|\{i:\alpha_{j,\ell}^i>0\}| \ge n/2$.
		\item $\max_{i,j,\ell}|\alpha_{j,\ell}^i|\le \alpha_{max}$, where  $\alpha_{max}=2g_{\infty}(F) =\OFPT(1)$
		\item $\|\vev_j^*\|_{\infty}=\OFPT(1)$.
	\end{itemize}
\end{claim}

\begin{proof}[Proof of Claim~\ref{claim:infbound}]
Consider an arbitrary $\vev_j$ such that $\TwoBlock[]{\veg_j^0(F)}{\vev_j}\in\G([B,A])$. We have $A(\veg_j^i(F)-\vev_j)=0$ for every $1\le i\le n$, hence there exist $\bar{\alpha}_{j,\ell}^i\in\Z_+$, $\veg_\ell(A)\in\G(A)$ and $\veg_{\ell}(A)\sqsubseteq \veg_j^i(F)-\vev_j$ such that
$$\veg_{j}^i(F)-\vev_j=\sum_{\ell=1}^m \bar{\alpha}_{j,\ell}^i\veg_\ell(A), \quad \forall 1\le i\le n.$$
Note that $\|\TwoBlock[]{\veg_j^0(F)}{\vev_j}\|_{\infty}\le g_{\infty}(F)=\alpha_{max}/2$, consequently $\|\veg_j^i(F)-\vev_j\|_{\infty}\le \alpha_{max}$, and $\bar{\alpha}_{j,\ell}^i\le \alpha_{max}$. 
Consider the cardinality of the set $\{i:\bar{\alpha}_{j,\ell}^i>0\}$. If $1\le |\{i:\alpha_{j,\ell}^i>0\}|\le \lfloor n/2\rfloor$, we say $\ell$ is {\em unbalanced} for $\veg_j(F)$. Let $\bar{\alpha}_{j,max}^{i}=\max_{1\le i\le n}\bar{\alpha}_{j,\ell}^i$ and $UB_j$ be the set of all unbalanced indices $\ell$, we define 
$$\vev_j^*:=\vev_j+\sum_{\ell\in UB_j}\bar{\alpha}_{j,\max}^{i}\veg_\ell(A),$$ 
then 
$$\veg_{j}^i(F)-\vev_j^*=\sum_{\ell\in\{1,2,\cdots,m\}\setminus UB_j} \bar{\alpha}_{j,\ell}^i\veg_\ell(A)+\sum_{\ell\in UB_j}(\bar{\alpha}_{j,\max}^{i}-\bar{\alpha}_{j,\ell}^i)\cdot (-\veg_\ell(A)), \quad \forall 1\le i\le n.$$
Note that $-\veg_{\ell}(A)\in\G(A)$. For all the $\veg_{\ell}(A)$'s in $\G(A)$ that do not appear in the above equation, we take their coefficients as $0$. Furthermore, we have $|\bar{\alpha}_{j,\ell}^i|\le\alpha_{max}$ and $|\bar{\alpha}_{j,\max}^{i}-\bar{\alpha}_{j,\ell}^i|\le \alpha_{max}$ for all $i,\ell$. As $\|\vev_j\|_{\infty}=\OFPT(1)$, $\|\veg_{\ell}(A)\|_{\infty}=\OFPT(1)$, we know that $\|\vev_j^*\|_{\infty}=\OFPT(1)$.  Thus, the claim is proved.
\end{proof}


We call $(\veg_j^0(F),\vev_j^*,\vev_j^*,\cdots,\vev_j^*)$ as a canonical vector (of $\veg_j(F)$). Since $\|\vev_j^*\|_{\infty}=\OFPT(1)$ and $\|\veg_j^0(F)\|_{\infty}=\OFPT(1)$, there are at most $\tau=\OFPT(1)$ different kinds of canonical vectors. This means, there may be different $\veg_k(F)$'s with the same canonical vector. We list all the $\tau$ possible canonical vectors and let  $\ver_j:=(\vep_j^*,\vev_j^*,\vev_j^*,\cdots,\vev_j^*)$ be the $j$-th one. Let $CA_j$ be the set of indices of all $\veg_k(F)$'s whose canonical vector is $\ver_j$, then we have
\begin{eqnarray}\label{eq:inftybound-1}
\veg=\sum_{j=1}^\tau(\sum_{k\in CA_j}\alpha_k)\ver_j+\sum_{j=1}^{\tau}\sum_{k\in CA_j}\alpha_k\left(\veg_k(F)-\ver_j\right).
\end{eqnarray}

We say an $n$-dimensional vector $\vealpha=(\alpha^1,\alpha^2,\cdots,\alpha^n)\in\Z_{\ge 0}^n$ is {\em balanced}, if $\vealpha=0$, or $\|\vealpha\|_{\infty}\le \alpha_{max}=\OFPT(1)$ and $|\{i:\alpha^i>0\}|\ge n/2$. Then the following observation is true.
\begin{observation}\label{obs:1}
	For any nonzero balanced vector $\vealpha$ it holds that $\|\vealpha\|_1\ge n/2\cdot \alpha^i/\alpha_{max}$ for every $1\le i\le n$.
\end{observation}

Using the concept of a balanced vector, Claim~\ref{claim:infbound} indicates that if $\ver_j$ is a canonical vector of $\veg_k(F)$, then $\veg_{k}^i(F)-\vev_j^*=\sum_{\ell=1}^{|\G(A)|} \alpha_{j,\ell}^i\veg_\ell(A)$ such that the vector $(\alpha_{j,\ell}^1,\alpha_{j,\ell}^2,\cdots,\alpha_{j,\ell}^n)$ is a balanced vector.

The nice thing about balanced vectors is that we can have the following claim, which will be used several times later.
\begin{claim}\label{claim:inftybound-2}
	Let $\vey_1,\vey_2,\cdots,\vey_{k}$ be a sequence of balanced vectors in $\Z_{\ge 0}^{n}$ 
	such that $\|\sum_{h=1}^k\vey_h\|_1\le n\Lambda$ where $\Lambda=\OFPT(1)$, then $\|\sum_{h=1}^k\vey_h\|_{\infty}\le 2\alpha_{max}\Lambda=\OFPT(1)$.
\end{claim}
\begin{proof}[Proof of Claim~\ref{claim:inftybound-2}]
	We prove by contradiction. Suppose on the contrary that $\|\sum_{h=1}^k\vey_h\|_{\infty}>2\alpha_{max}\Lambda$, then there exists some $i^*$ such that $\sum_{h=1}^k\vey_h^{i*}>2\alpha_{max}\Lambda$. Since $\vey_h$'s are balanced vectors, according to Observation~\ref{obs:1}, we have
	$$\|\sum_{h=1}^k\vey_h\|_1=\sum_{h=1}^k\|\vey_h\|_1\ge n\cdot \frac{\sum_{h=1}^k\vey_h^{i*}}{2\alpha_{max}}>n\Lambda,$$
	which contradicts the fact that $\|\sum_{h=1}^k\vey_h\|_1\le n\Lambda$. Hence, the claim is true.
\end{proof}

Since $\ver_j$ is a canonical vector of $\veg_k(F)$, by Claim~\ref{claim:infbound}, there exist balanced vectors $\vebeta_{k,\ell}$ such that Eq~(\ref{eq:inftybound-1}) can be rewritten as 
\begin{eqnarray*}
\veg^i=\sum_{j=1}^\tau(\sum_{k\in CA_j}\alpha_k)\vev_j^*+\sum_{j=1}^{\tau}\sum_{k\in CA_j}\alpha_k\left(\sum_{\ell=1}^{|\G(A)|}\beta_{k,\ell}^i\veg_\ell(A)\right), \quad \forall 1\le i\le n,
\end{eqnarray*}
or equivalently,
\begin{eqnarray}\label{eq:inftybound-2}
	\veg^i=\sum_{j=1}^\tau\alpha_j'\vev_j^*+\sum_{\ell=1}^{|\G(A)|}\beta^i_\ell\veg_\ell(A),\quad \forall 1\le i\le n,
\end{eqnarray}
where $\alpha_j'=\sum_{k\in CA_j}\alpha_k$ and each vector $\beta_\ell=(\beta_\ell^1,\cdots,\beta_\ell^n)$ is the summation of some balanced vectors.

As $[0,D,D,\cdots,D]\veg=0$, we have
\begin{eqnarray}\label{eq:inftybound-3}
\sum_{j=1}^\tau n\alpha_j'D\vev_j^*+\sum_{\ell=1}^{|\G(A)|}(\sum_{i=1}^n\beta_\ell^i)\veg_\ell(A)=0.
\end{eqnarray}
Note that $|\G(A)|=\OFPT(1)$, the equation above can be rewritten as 
\begin{eqnarray}\label{eq:inftybound-5}
[D\vev_1^*,D\vev_2^*,\cdots,D\vev_\tau^*, \veg_1(A),\veg_2(A),\cdots,\veg_{|\G(A)|}(A)]\cdot (n\alpha_1',n\alpha_2',\cdots,n\alpha_\tau',\sum_{i=1}^n\beta_{1}^i,\cdots,\sum_{i=1}^n\beta^i_{|\G(A)|})=0.
\end{eqnarray}

Let $V=[D\vev_1^*,D\vev_2^*,\cdots,D\vev_\tau^*, \veg_1(A),\veg_2(A),\cdots,\veg_{|\G(A)|}(A)]$, 
there exist $\lambda_k\in\Z_+$ and $\veg_k(V)\in \G(V)$, such that
$$(n\alpha_1',n\alpha_2',\cdots,n\alpha_\tau',\sum_{i=1}^n\beta_{1}^i,\cdots,\sum_{i=1}^n\beta^i_{|\G(A)|})=\sum_{k}\lambda_k \veg_k(V).$$
Note that since $\alpha_j',\beta_\ell^i\ge 0$, we can restrict that every $\veg_k(V)\in\Z_{\ge 0}^{\tau+|\G(A)|}$. 

There are two possibilities regarding the values of $\lambda_k$'s.

\smallskip
\noindent\textbf{Case 1.} $\lambda_k<n$ for every $k$. In this case we prove that $\|\veg\|_{\infty}=\OFPT(1)$ and Theorem~\ref{lemma:3-infty-bound} follows directly. Note that $V$ is a matrix of $\OFPT(1)$ size with $\|V\|_{\infty}=\OFPT(1)$, hence $\|\veg_k(V)\|_{\infty}=\OFPT(1)$ and $|\G(V)|=\OFPT(1)$. Therefore, $n\alpha_j'<n|\G(V)|\cdot \max_k\|\veg_k(V)\|_{\infty}=\OFPT(n)$, implying that $\alpha_j'=\OFPT(1)$. Consider the vector $\vebeta=(\beta^1,\cdots,\beta^n)$ where $\beta^i=\sum_{\ell }\beta_{\ell}^i$. As $\beta_{\ell}^i\ge 0$, 
$$\|\vebeta\|_1=\|(\sum_{i=1}^n\beta_1^i,\cdots,\sum_{i=1}^n\beta_{|\G(A)|}^i)\|_1\le \sum_k\lambda_k\|\veg_k(V)\|_1=\OFPT(n).$$
Recall that $\vebeta_{\ell}$ is the summation of balanced vectors, whereas $\vebeta=\sum_{\ell}\vebeta_{\ell}$ is also the summation of balanced vectors. Using Claim~\ref{claim:inftybound-2}, $\|\vebeta\|_{\infty}=\OFPT(1)$.

Combining the fact that $\|\vev_j^*\|_{\infty}=\OFPT(1)$ and $\|\veg_\ell(A)\|_{\infty}=\OFPT(1)$, we conclude that $\|\veg^i\|_{\infty}=\OFPT(1)$ for $1\le i\le n$. Meanwhile, as $\veg^0=\sum_j\alpha_j'\vep_{j}^*$, we have $\|\veg^0\|_{\infty}=\OFPT(1)$. Hence, $\|\veg\|_{\infty}=\OFPT(1)$.

\smallskip
\noindent\textbf{Case 2.} $\lambda_k\ge n$ for some $k$. For ease of description, we take the viewpoint of a packing problem. We view
each canonical vector $\ver_j^*$ and $\veg_\ell(A)$ as an item, whereas there are $\tau+|\G(A)|$ different kinds of items. There are $n+1$ different bins. Bin $0$ can only be used to pack items $\ver_j^*$, $1\le j\le \tau$, and bin $i$ ($1\le i\le n$) can only be used to pack items $\veg_\ell(A)$, $1\le \ell\le |\G(A)|$. Currently there are $\alpha_j'$ copies of item $\ver_j^*$ in bin $0$, and $\beta^i_\ell$ copies of item $\veg_\ell(A)$ in bin $i$. This is called a packing profile. Now we want to split this packing profile into several \lq\lq sub-profiles\rq\rq, i.e., we want to determine integers $\mu_j^h,\sigma^{i,h}_{\ell}\in\Z_{\ge 0}$ such that the followings are true:
\begin{enumerate}[(i)]
	\item $\mu_j^h, \sigma^{i,h}_{\ell}=\OFPT(1)$ and $\mu_j^h+ \sigma^{i,h}_{\ell}>0$.
	\item $\sum_h \mu_j^h=\alpha_j'$, $\sum_{h}\sigma^{i,h}_{\ell}=\beta^i_\ell$;
	\item $[D\vev_1^*,D\vev_2^*,\cdots,D\vev_\tau^*, \veg_1(A),\veg_2(A),\cdots,\veg_{|\G(A)|}(A)]\cdot (n\mu_1^h,n\mu_2^h,\cdots,n\mu_\tau^h,\sum_{i=1}^n\sigma^{i,h}_{\ell},\cdots,\sum_{i=1}^n\sigma^{i,h}_{|\G(A)|})=0$ for every $h$.
\end{enumerate}
A packing with $\mu_j^h$ copies of $\ver_j^*$ in bin $0$ and $\sigma^{i,h}_\ell$ copies of $\veg_\ell(A)$ in bin $i$ is called a sub-profile. Any sub-profile corresponds to a $(t_A+nt_B)$-dimensional vector $\vece_h=(\vece_h^0,\vece_h^1,\cdots,\vece_h^n)$ where
\begin{flalign*}
&\vece_h^0=\sum_{j=1}^\tau \mu_j^h\vep_j^* \\
&\vece_h^i=\sum_{j=1}^\tau \mu_j^h\vev_j^*+\sum_{\ell=1}^{|\G(A)|}\sigma_\ell^{i,h}\veg_\ell(A), \quad\forall 1\le i\le n
\end{flalign*}
If all the three conditions on sub-profiles hold, then we know that $\|\vece_h\|_{\infty}=\OFPT(1)$, $\veg=\sum_h \vece_h$ and $H_0\vece_h=0$ (to see why $H_0\vece_h=0$ holds, simply observe that $F\ver_j^*=0$ and condition (iii) implies that $[0,D,D,\cdots,D]\vece_h=0$), and furthermore, there are at most $\sum_j\alpha_j'+\sum_{i,\ell}\beta_\ell^i$ sub-profiles, which is finite. Hence, $\veg=\sum_h \vece_h$ and the theorem is proved.

We will construct $\vece_h$'s iteratively. Once $\vece_h$ is constructed, we continue our decomposition procedure on $\veg-\sum_{k=1}^h\vece_k$. 

Suppose we have constructed $\vece_1$ to $\vece_{h_0-1}$ where conditions (i) and (iii) are satisfied for each $\vece_h$, $\alpha_j'-\sum_{h=1}^{h_0-1}\mu_j^h\ge 0$, $\bar{\beta}_\ell^i:=\beta_\ell^i-\sum_{h=1}^{h_0-1}\sigma_\ell^{i,h}\ge 0$ and furthermore, each vector $\bar{\vebeta}_\ell=(\bar{\beta}_\ell^1,\cdots,\bar{\beta}_\ell^n)$ can be expressed as a summation of all but one balanced vectors, more precisely, there exist balanced vectors $\phi_{\ell,k}\in\Z_{\ge 0}^n$, $1\le k\le k_{max}$ such that
$$\bar{\beta}_\ell=\sum_{k=1}^{k_{max}-1}\phi_{\ell,k}+\bar{\phi}_{\ell,k_{max}}$$
where $\bar{\phi}_{\ell,k_{max}}\sqsubseteq \phi_{\ell,k_{max}}$. 

We show how to construct $\vece_{h_0}$. Let $\bar{\alpha}_j'=\alpha_j'-\sum_{h=1}^{h_0-1}\mu_j^h$. According to condition (iii) of each $\vece_h$, we know that
$$[D\vev_1^*,D\vev_2^*,\cdots,D\vev_\tau^*, \veg_1(A),\veg_2(A),\cdots,\veg_{|\G(A)|}(A)]\cdot (n\bar{\alpha}_1',n\bar{\alpha}_2',\cdots,n\bar{\alpha}_\tau',\sum_{i=1}^n\bar{\beta}_{1}^i,\cdots,\sum_{i=1}^n\bar{\beta}^{i}_{|\G(A)|})=0$$
Consequently, there exist $\lambda_k'\in\Z_{\ge 0}$ and $\veg_k\in\Z_{\ge 0}^{\tau+|\G(A)|}\cap \G(V)$ such that
$$(n\bar{\alpha}_1',n\bar{\alpha}_2',\cdots,n\bar{\alpha}_\tau',\sum_{i=1}^n\bar{\beta}_{1}^i,\cdots,\sum_{i=1}^n\bar{\beta}^{i}_{|\G(A)|})=\sum_{k}\lambda_k' \veg_k(V).$$
There are two possibilities. 

\smallskip
\noindent\textbf{Case 2.1 } If there exists some $\lambda_k'\ge n$, we consider the vector-summand $n\veg_k(V)$ out of $\lambda_k'\veg_k(V)$. Let $n\veg_k(V)=(n\zeta_1,n\zeta_2,\cdots,n\zeta_{\tau+|\G(A)|})$. We set $\mu_{j}^{h_0}=\zeta_j=\OFPT(1)$ for $1\le j\le \tau$. We set the values of $\sigma_{\ell}^{i,h_0}$ such that $\sum_{i=1}^n\sigma_\ell^{i,h_0}=n\zeta_{\tau+\ell}$. Consequently, condition (iii) is satisfied for $\vece_{h_0}$. Now it suffices to set the values of each $\sigma_{\ell}^{i,h_0}$ such that they are bounded by $\OFPT(1)$. Equivalently, this means out of the $\bar{\beta}_\ell^i$ copies of $\veg_\ell(A)$, our goal is to take $\sigma_{\ell}^{i,h_0}$ copies such that in total we take $n\zeta_{\tau+\ell}$ copies and $\sigma_{\ell}^{i,h_0}=\OFPT(1)$. We achieve this in a simple greedy way. Let $k^*$ be the index such that
$$\sum_{k=k^*+1}^{k_{max}-1}\|\phi_{\ell,k}\|_1+\|\bar{\phi}_{\ell,k_{max}}\|_1<n\zeta_{\tau+\ell}\le \sum_{k=k^*}^{k_{max}-1}\|\phi_{\ell,k}\|_1+\|\bar{\phi}_{\ell,k_{max}}\|_1$$
Let $\bar{\phi}_{\ell,k^*}\sqsubseteq \phi_{\ell,k^*}$ be an arbitrary vector such that
$$\|\bar{\phi}_{\ell,k^*}\|_{1}+\sum_{k=k^*+1}^{k_{max}-1}\|\phi_{\ell,k}\|_1+\|\bar{\phi}_{\ell,k_{max}}\|_1=n\zeta_{\tau+\ell}.$$
We set $\sigma_{\ell}^{i,h_0}=\bar{\phi}_{\ell,k^*}^i+\sum_{k=k^*+1}^{k_{max}-1}\phi_{\ell,k}^i+\bar{\phi}_{\ell,k_{max}}^i$. It is obvious that in total we have taken $n\zeta_{\tau+\ell}$ copies of $\veg_\ell(A)$. Now it remains to show that $\|\sigma_\ell^{h_0}\|_{\infty}=\|\bar{\phi}_{\ell,k^*}+\sum_{k=k^*+1}^{k_{max}-1}\phi_{\ell,k}+\bar{\phi}_{\ell,k_{max}}\|_{\infty}=\OFPT(1)$. To see this, notice that each $\phi_{\ell,k}$ is a balanced vector, hence $$\|{\phi}_{\ell,k^*}\|_{1}+\sum_{k=k^*+1}^{k_{max}-1}\|\phi_{\ell,k}\|_1+\|{\phi}_{\ell,k_{max}}\|_1\le n\zeta_{\tau+\ell}+2n\alpha_{max}=\OFPT(n).$$
According to Claim~\ref{claim:inftybound-2}, $\|{\phi}_{\ell,k^*}+\sum_{k=k^*+1}^{k_{max}-1}\phi_{\ell,k}+{\phi}_{\ell,k_{max}}\|_{\infty}=\OFPT(1)$. Consequently, $\|\sigma_\ell^{h_0}\|_{\infty}=\OFPT(1)$.

Also notice that after we take $\sigma_{\ell}^{i,h_0}$ copies of $\veg_\ell(A)$, 
$$\bar{\vebeta}_\ell-\sigma_{\ell}^{h_0}=\sum_{k=1}^{k^*-1}\phi_{\ell,k}+(\phi_{\ell,k^*}-\bar{\phi}_{\ell,k^*}),$$ 
which is still the summation of all but one balanced vector. Hence we can continue to decompose $\veg-\sum_{h=1}^{h_0}\vece_h$.

\smallskip
\noindent\textbf{Case 2.2 } $\lambda_k'< n$ for every $k$. We claim that $\|\veg-\sum_{h=1}^{h_0-1}\vece_h\|_{\infty}={\OFPT(1)}$. If this claim is true, then $\veg=\sum_{h=1}^{h_0-1}\vece_h+(\veg-\sum_{h=1}^{h_0-1}\vece_h)$, and Theorem~\ref{lemma:3-infty-bound} is proved. To show the claim, we use a similar argument as that of case 1. First, $n\bar{\alpha}_j'\le (\sum_k\lambda_k)\cdot \max_k\|\veg_k(V)\|_{\infty}=\OFPT(n)$, hence $\bar{\alpha}_j'=\OFPT(1)$. Second, we consider the $n$-dimensional vector $\vebeta=\sum_{\ell=1}^{|\G(A)|}\vebeta_{\ell}$. Recall that $\bar{\beta}_\ell^i:=\beta_\ell^i-\sum_{h=1}^{h_0-1}\sigma_\ell^{i,h}\ge 0$ and each vector $\bar{\vebeta}_\ell$ satisfy that
$$\bar{\beta}_\ell=\sum_{k=1}^{k_{max}-1}\phi_{\ell,k}+\bar{\phi}_{\ell,k_{max}}$$
where $\bar{\phi}_{\ell,k_{max}}\sqsubseteq \phi_{\ell,k_{max}}$.  Let $\bar{\beta}_\ell'=\sum_{k=1}^{k_{max}}\phi_{\ell,k}$ and $\vebeta'=\sum_{\ell=1}^{|\G(A)|}\vebeta_{\ell}'$. Given that $\bar{\phi}_{\ell,k_{max}}\sqsubseteq {\phi}_{\ell,k_{max}}$ and ${\phi}_{\ell,k_{max}}$ is a balanced vector, $\|\bar{\beta}_\ell'\|_1\le \|\bar{\beta}_\ell\|_1+n\alpha_{max}$. Consequently
$$\|\beta'\|_1\le \sum_{\ell=1}^{|\G(A)|}\|\bar{\beta}_\ell'\|_1\le \sum_{\ell=1}^{|\G(A)|}\|\bar{\beta}_\ell\|_1+n\alpha_{max}\cdot |\G(A)|\le \sum_k\lambda_k'\cdot \max_k\|\veg_k(V)\|_1+n\alpha_{max}\cdot |\G(A)|=\OFPT(n).$$
Note that $\vebeta'$ is the summation of balanced vectors. According to Claim~\ref{claim:inftybound-2}, $\|\vebeta'\|_{\infty}=OFPT(1)$, consequently $\|\vebeta\|_{\infty}\le \|\vebeta'\|_{\infty}=OFPT(1)$.
Combining the fact that $\|\vep_{j}^*\|_{\infty}=\OFPT(1)$, $\|\vev_j^*\|_{\infty}=\OFPT(1)$ and $\|\veg_\ell(A)\|_{\infty}=\OFPT(1)$, we conclude that $\|\veg-\sum_{i=1}^{h_0-1}\vece_h\|_{\infty}=\OFPT(1)$. 
\end{proof}

Theorem~\ref{lemma:3-infty-bound} indicates that, there exists some ``basis'' for 3-block $n$-fold IP with FPT-bounded $\ell_{\infty}$-norms. Unfortunately, this basis need not be Graver basis; indeed, we will show later that the Graver basis of 3-block $n$-fold IP does not have a bounded $\ell_{\infty}$-norm.
However, Theorem~\ref{lemma:3-infty-bound} provides a new perspective on the structure of the kernel space, which can be utilized to bound the $\ell_{\infty}$-norm of the Graver basis through a ``merging'' technique as we illustrate in the following subsection.

\subsection{A sign-compatible decomposition}\label{subsec:3block-graver-decompose}
We have shown in the previous subsection that any element of $\ker_{\Z}(H_0)$ admits a decomposition into lattice elements whose $\ell_{\infty}$-norm is bounded by $\OFPT(1)$. However, this decomposition is not necessarily ``sign-compatible'', meaning that possibly none of its elements is a feasible step on its own, which makes its immediate algorithmic use complicated. Towards the algorithm for 3-block $n$-fold IP, we resort to Graver basis. The goal of this subsection is to prove the following theorem.



\begin{reptheorem}{thm:3-block-graver}[\apx]
For any $3$-block $n$-fold matrix $H_0$, $g_{\infty}(H_0)\le 
\OFPT(n^{t_A^2+1})$.	
\end{reptheorem}

Following the line of arguments in previous papers~\cite{aschenbrenner2007finiteness,hemmecke2014graver,hemmecke2011polynomial,hocsten2007finiteness}, it seems very difficult to derive an upper bound singly exponential in $t_A$. To prove the Theorem~\ref{thm:3-block-graver}, we use a completely different approach. We give a brief overview of the proof idea, and the reader is referred to Appendix~\ref{appsec:3block-graver-decompose} for details. 

\smallskip
\noindent\textbf{Proof idea.} The basic idea is to show that if $\|\veg(H_0)\|_{\infty}$ is too large for some $\veg(H_0)\in\G(H_0)$, then we are able to find some $\vez\sqsubset \veg(H_0)$ and $H_0\vez=0$, contradicting the fact that $\veg(H_0)$ is a Graver basis element.
Suppose $\vey=\veg(H_0)$ and $\|\vey\|_{\infty}$ is very large.
The crucial idea is that we do not search directly for $\vez\sqsubset \vey$, but rather search for $\vez\sqsubset \tilde{\vey}$ where $\tilde{\vey}$ is an ``equalization'', of $\vey$, and then prove that such a $\vez$ also satisfies that $\vez\sqsubset \vey$.

Roughly speaking, we will divide the $n$ bricks of $\vey$, i.e., $\vey^i$ for $i=1,2,\cdots, n$, into $\sigma=\OFPT(1)$ groups $N_1,N_2,\cdots,N_{\sigma}$ such that for any $k\in N_j$, $\tilde{\vey}^k\approx \frac{1}{|N_j|} \sum_{i\in N_j}\vey^i$.  Why do we need to take such a detour in the proof? The problem is that by directly arguing on $\vey$ we run into a bound that is similar as~\cite{hemmecke2014graver}.
Therefore, we use a completely different approach -- we adopt the decomposition of Theorem~\ref{lemma:3-infty-bound}, and then modify such a decomposition into a sign-compatible one by ``merging'' summands. Towards this, we first prove a merging lemma (Lemma~\ref{lemma:merging-lemma}) which states that given a summation of a sequence of vectors, we can always divide them into disjoint subsets such the summation of vectors in each subset becomes sign-compatible. The merging lemma can turn an arbitrary decomposition into a sign-compatible one, despite the fact that the cardinality of each subset is exponential in the dimension of the vectors (which means the $\ell_{\infty}$-norm of the summands will explode by a factor that is exponential in the dimension). Consequently, if we directly merge the $\OFPT(n)$-dimensional vectors in the decomposition of Theorem~\ref{lemma:3-infty-bound}, we get a very weak bound. To handle this, we consider an alternative sum $\tilde{\vey}$, which is derived by averaging multiple bricks of $\vey$ as we mentioned above.

By altering the decomposition of $\vey$, we get a decomposition of $\tilde{\vey}$ such that the following is true: all the $n+1$ bricks of each vector-summand can be divided into $\OFPT(1)$ subsets where in each subset the bricks are identical. This indicates that, although we are summing up $\OFPT(n)$-dimensional vectors to $\tilde{\vey}$, it is essentially the same as summing up $\OFPT(1)$-dimensional vectors. Such a transformation comes at a cost -- summands summing up to $\tilde{\vey}$ do not have $\OFPT(1)$-bounded $\ell_{\infty}$-norms, indeed, each vector-summand consists of $n$ bricks whose $\ell_{\infty}$-norm is $\OFPT(1)$, and at most $1$ brick (which is a $t_A$-dimensional vector) whose $\ell_{\infty}$-norm is $\OFPT(n)$. Applying our merging lemma, we derive a sign-compatible decomposition of $\tilde{\vey}$ where the summands have an $\ell_{\infty}$-norm bounded by $\OFPT(n^{t_A^2+1})$.

It remains to show that at least one summand $\vez$ in the sign-compatible decomposition of $\tilde{\vey}$ also satisfies that $\vez\sqsubset \vey$. To show this we need to go back to the definition of $\tilde{\vey}$. We are averaging bricks of $\vey$, but which bricks shall we average? Each brick is a $t_A$-dimensional vector and we consider each coordinate. We set up a threshold $\Gamma$. If the absolute value of a coordinate is larger than $\Gamma$, we say it is (positive or negative) large. Otherwise it is small. Therefore, each brick can be characterized by identifying its coordinates being positive large, negative large or small (which is defined as the \emph{quantity type} of a brick). We only average the bricks of the same quantity type. By doing so, we can ensure that $\tilde{\vey}^i$ is roughly sign-compatible with $\vey^i$ -- if the $j$-th coordinate of $\vey^i$ is positive or negative large, then this coordinate of $\tilde{\vey}^i$ is also positive or negative. Hence, any $\vez \sqsubset \tilde{\vey}^i$ is almost sign-compatible with $\vey$ -- indeed, if we can ensure additionally that the $j$-th coordinate of $\vez^i$ is $0$ as long as the $j$-th coordinate of $\vey^i$ is small, then we can conclude that $\vez\sqsubset \vey$. This \lq\lq if\rq\rq\, can be proved using a counting argument, and we get Theorem~\ref{thm:3-block-graver}.





\subsection{$4$-block $n$-fold IP reduces to $3$-block $n$-fold IP}
In this subsection, we will show that for any $4$-block $n$-fold IP, there exists an equivalent $3$-block $n$-fold IP which is {\em kernel preserving}, as we define in the following.


\begin{definition}[Extended formulation] \label{def:ef}
Let $n' \geq n$, $m' \in \N$, $\A \in \Z^{m \times n}$, $\veb \in \Z^m$, $\vel, \veu \in (\Z \cup \{\pm \infty\}^n$ and $\A' \in \Z^{m^{\prime} \times n^{\prime}}$, $\veb' \in \Z^{m'}$, $\vel', \veu' \in (\Z \cup \{\pm \infty\})^{n'}$.
We say that 
\begin{equation}
\A'(\vex, \vey) = \veb',\, \vel' \leq (\vex, \vey) \leq \veu' \tag{EF} \label{EF}
\end{equation} is an \emph{extended formulation} of 
\begin{equation}
\A \vex = \veb,\, \vel \leq \vex \leq \veu \tag{OrigF} \label{OrigF}
\end{equation}
if $\{\vex \mid \A \vex = \veb, \vel \leq \vex \leq \veu\} = \{\vex \mid \exists \vey: \A'(\vex, \vey) = \veb', \vel' \leq (\vex, \vey) \leq \veu'\}$.
\end{definition}

\begin{definition}[Feasibly kernel-preserving extended formulation] \label{def:fkp-ef}
We say that~\eqref{EF} is a \emph{feasibly kernel preserving} extended formulation of~\eqref{OrigF} if for each $(\vex, \vey)$ feasible in~\eqref{EF},
$$\A'(\veg, \veh) = \vezero, \, \vel' \leq (\vex, \vey) + (\veg, \veh) \leq \veu' \qquad \implies \qquad \A \veg = \vezero, \, \vel \leq \vex + \veg \leq \veu,$$
that is, each element $(\veg, \veh)$ of $\ker(\A')$ which is feasible with respect to $(\vex, \vey)$ corresponds to an element $\veg \in \ker(\A)$ which is feasible with respect to $\vex$.
\end{definition}

Extended formulations are commonly used to show how a set of solutions can be embedded in an \emph{extended space}, perhaps using less inequalities or obeying some extra structural requirements.
The basic observation is that if we take an objective function $f$ over the original formulation~\eqref{OrigF} and optimize $f'(\vex, \vey) = f(\vex)$ over~\eqref{EF}, the optimal solution $(\vex, \vey)$ over~\eqref{EF} is such that $\vex$ is an optimum over~\eqref{OrigF}.
In the subsequent theorem we will use it to show that any $4$-block $n$-fold IP can be embedded in a $3$-block $n$-fold IP without blowing up the block sizes too much.
The specific notion of a feasibly kernel-preserving extended formulation is useful to show that also our lower bounds on lattice elements are transferred, as we will show subsequently in Theorem~\ref{thm:3-block-better-lower}.

Now we come to the main result of this subsection.
\begin{theorem}[\apx]\label{thm:4block-3block}
Any $4$-block $n$-fold IP with parameters $\Delta, s_A, s_B, s_C, s_D, t_A, t_B, t_C, t_D$ has a feasibly kernel-preserving extended formulation whose constraint matrix is a $3$-block $n$-fold matrix with parameters $\hat{\Delta}, \hat{s}_A, \hat{s}_B, \hat{s}_D, \hat{t}_A, \hat{t}_B, \hat{t}_D$ satisfying 
\begin{align*}
\hat{\Delta} &= \Delta & \hat{t}_A &= \hat{t}_D =  2t_C + t_D + s_A  & \hat{t}_B &= t_B & 
\hat{s}_A = \hat{s}_B = s_B + t_C && \hat{s}_D &= s_D = s_C \enspace .
\end{align*}
\end{theorem}

\noindent\textbf{Remark.} Theorem~\ref{thm:4block-3block} has several consequences. One is that 4-block $n$-fold IP is in FPT if and only if 3-block $n$-fold IP is in FPT. Furthermore, as the reduction is kernel preserving, we can also utilize Theorem~\ref{thm:4block-3block} to transfer the Graver basis elements between 4-block $n$-fold IP and 3-block $n$-fold IP, 
as is implied by the following theorem.

\begin{reptheorem}{thm:3-block-better-lower}[\apx]
There exists an instance of 3-block $n$-fold IP with a matrix $H_0$ 
where $\hat{s}_A=\hat{s}_B=2t, \hat{s}_D=t-1,\hat{t}_A=\hat{t}_D=4t$, $\hat{t}_B=t$ such that for every feasible solution $\vex$ and for every $\veg \in \ker_{\Z}(H_0)$ which is feasible with respect to $\vex$, it holds that $\|\veg\|_\infty = \Omega(n^{t-1})$, and in particular, $\|\veg^0\|_\infty = \Omega(n^{t-1})$. 
\end{reptheorem}

\section{Algorithms}


Using the upper bound on the Graver basis elements, we can derive algorithms for $3$-block and 4-block $n$-fold IP by combining the idea from~\cite{hemmecke2014graver} and the recent progress in~\cite{martin2018parameterized,eisenbrand2018faster}.

\begin{reptheorem}{thm:alg-3-block}[\apx]
	There exists an algorithm for 3-block $n$-fold IP that runs in $\min\{\OFPT(n^{s_{\scalebox{.5}{D}}t_B+3}\log^3 n), \OFPT(n^{(t_A^2+1)t_B+3}\log^3 n)\}$ time.
\end{reptheorem}


Using the same argument and Theorem~\ref{thm:3-block-graver-4}, we have the following.
\begin{reptheorem}{thm:alg-4-block}[\apx]
	There exists an algorithm for 4-block $n$-fold IP that runs in $\OFPT(n^{s_{\scalebox{.5}{D}}t_B+3}\log^3 n)$ time.	
\end{reptheorem}

\section{Conclusion}
We consider 4-block $n$-fold IP and its important special case $3$-block $n$-fold IP, both generalizing the well-known two-stage stochastic IP and $n$-fold IP. We show that lattice elements of $3$-block $n$-fold IP admits a lattice decomposition whose $\ell_{\infty}$-norm is bounded in $\OFPT(1)$, while any non-zero integral element in the kernel space of $4$-block $n$-fold IP may have an $\ell_{\infty}$-norm at least $\Omega(n^{s_c})$. We provide a matching upper bound on the $\ell_{\infty}$-norm of the Graver basis for $4$-block $n$-fold IP, which gives an exponential improvement upon the best known result. We also establish an upper bound of $\min\{\OFPT(n^{s_c}),\OFPT(n^{t_A^2}+1)\}$ on the $\ell_{\infty}$-norm of the Graver basis for $3$-block $n$-fold IP. 

It remains as an important open problem whether $4$-block $n$-fold IP, or equivalently, $3$-block $n$-fold IP, is in FPT. Our lower bound results indicate that, following the classical iterative augmentation framework, it is unlikely to derive an FPT algorithm. 

\clearpage
\appendix

\section{Proof of Theorem~\ref{thm:4-block-lower}}
\begin{reptheorem}{thm:4-block-lower}
	There exists a $4$-block $n$-fold matrix $H$ such that $s_C=s_D=t-1$, $t_C=t_D=t$, $s_A=s_B=t_A=t_B=t$, and for any nonzero $\vey\in \ker_{\Z}(H)$, $\|\vey\|_{\infty}=\Omega({n^{t-1}})$, and in particular, $\|\vey^0\|_\infty = \Omega(n^{t-1})$.	
	\end{reptheorem}
\begin{proof}
	We let $A=I_{t\times t}$, $B=-I_{t\times t}$. We define $(t-1)\times t$ matrices $D$ and $C$ such that
	\[ D=
	\begin{pmatrix}
	1 & -1 & 0   &\cdots &0 & 0\\
	0 & 1 & -1  &\cdots &0 & 0 \\
	\vdots  &  &  & \ddots & &   \\
	0  & 0 & 0  & \cdots & 1 & -1\\
	\end{pmatrix}
	\hspace{10mm} C=
	\begin{pmatrix}
	-1 & 0 & 0   &\cdots &0 & 0\\
	0 & -1 & 0  &\cdots &0 & 0 \\
	\vdots  &  &  & \ddots & &   \\
	0  & 0 & 0  & \cdots & -1 & 0\\
	\end{pmatrix}
	\]
	
	Consider any nonzero $\vey\in \ker_{\Z}(H)$. According to $A\vey^0-B\vey^i=0$, we know that $\vey^0=\vey^i$ for every $1\le i\le n$. According to $C\vey^0+\sum_{i=1}^n D\vey^i=0$, we have $(C+nD)\vey^0=0$, i.e.,
	\[
	\begin{pmatrix}
	n-1 & -n & 0   &\cdots &0 & 0\\
	0 & n-1 & -n  &\cdots &0 & 0 \\
	\vdots  &  &  & \ddots & &   \\
	0  & 0 & 0  & \cdots & n-1 & -n\\
	\end{pmatrix}
	\cdot \vey=0
	\]

	Let $\vey^0=(y_1,y_2,\cdots,y_t)$, the following is true:
	\begin{eqnarray}\label{eq:lower-bound-4-lock}
	(n-1)y_i=ny_{i+1},\quad 1\le i\le t-1
	\end{eqnarray}
	It is easy to see that as long as $\vey\neq 0$, we have $\vey^0\neq 0$ and consequently $y_i\neq 0$ for every $1\le i\le t$. 
	According to $(n-1)y_{t-1}=ny_t$, $y_{t-1}$ is dividable by $n$. Let $y_{t-1}=nz_{t-1}$ for some $z_{t-1}\in \mathbb{Z}_{\neq 0}$. According to $(n-1)y_{t-2}=ny_{t-1}=n^2z_{t-1}$, we know that $y_{t-2}$ is dividable by $n^2$. Let $y_{t-2}=n^2z_{t-2}$ and we plug it into $(n-1)y_{t-3}=ny_{t-2}$. In general, suppose we have shown that $y_{t-k}=n^kz_{t-k}$ for all $k\le k_0$. Now for $k=k_0+1$, we have $(n-1)y_{t-k_0-1}=ny_{t-k_0}=n^{k_0+1}z_{n-k_0}$, then $y_{t-k_0-1}$ is dividable by $n^{k_0+1}$. Hence, we conclude that $y_1$ is dividable by $n^{t-1}$, i.e., $\|\vey\|_{\infty}=\Omega(n^{t-1})$ and Theorem~\ref{thm:4-block-lower} is proved.
\end{proof}

\section{Proof of Theorem~\ref{thm:3-block-graver}}\label{appsec:3block-graver-decompose}
\begin{reptheorem}{thm:3-block-graver}
	For any $3$-block $n$-fold matrix $H_0$ and $\veg(H_0)\in \G(H_0)$, $\|\veg(H_0)\|_{\infty}\le 
	\OFPT(n^{t_A^2+1})$.	
\end{reptheorem}

\subsubsection{A merging lemma}\label{subsec:merging}

We show in this subsection a merging lemma that allows us to transform an arbitrary summation of vectors into a sign-compatible summation by merging the vector-summands.

We start with a merging lemma on 1-dimensional vectors to illustrate the main idea.
\begin{lemma}\label{lemma:merging-lemma-1}
	Let $x_1,x_2,\cdots,x_m$ be a sequence of integers such that $x=\sum_{i=1}^mx_i$, and $|x_i|\le \zeta$. Then the $m$ integers can be partitioned into $m'$ subsets $T_1,T_2,\cdots,T_{m'}$ satisfying that: $\cup_{j=1}^{m'}T_j=\{1,2,\cdots,m\}$, and for every $1\le j\le m'$ it holds that $\sum_{i\in T_j} x_i\sqsubseteq x$, $|T_j|\le  6\zeta+2$. 
\end{lemma}
\begin{proof}
	Without loss of generality we assume that $x\ge 0$ (otherwise we argue on $-x_i$'s). If $m\le 6\zeta+2$ the lemma is trivial. Otherwise we apply Steinitz lemma (Lemma~\ref{lemma:steinitz}) to the integral sequence $x_1,x_2,\cdots,x_m$ and there exists a permutation $\pi$ such that for all $1\le \ell\le m$ it holds that
	$$|\sum_{i=1}^\ell x_{\pi(i)}-\frac{\ell-1}{m}x|\le \zeta.$$
	Now we consider the first $3\zeta+2$ numbers $x_{\pi(1)},x_{\pi(2)},\cdots,x_{\pi(3\zeta+2)}$. There are two possibilities regarding $(3\zeta+1)x/m$. 
	
	If $(3\zeta+1)x/m> \zeta$, then since $-\zeta\le \sum_{i=1}^{3\zeta+2}x_{\pi(i)}-(3\zeta+1)x/m\le \zeta$, we know that $\sum_{i=1}^{3\zeta+2}x_{\pi(i)}\ge 0$, and is consequently sign-compatible with $x$. Further notice that the summation of the remaining integers satisfies that $\sum_{i=3\zeta+3}^m x_{\pi(i)}\ge x-(3\zeta+1)x/m-\zeta$. Given that $m\ge 6\zeta+2$, $x-(3\zeta+1)x/m\ge (3\zeta+1)x/m>\zeta$, the summation of the remaining integers is still positive.
	
	Otherwise $(3\zeta+1)x/m\le \zeta$, and consequently $0\le (\ell-1) x/m\le \zeta$ for any $1\le \ell\le 3\zeta+2$. This implies that the values of the $3\zeta+2$ numbers $\sum_{i=1}^{\ell}x_{\pi(i)}$ where $1\le \ell\le 3\zeta+2$ must lie in the set $\{-\zeta,-\zeta+1,\cdots,2\zeta\}$, i.e., there must exist two distinct integers $\ell_1<\ell_2$ and $\ell_1,\ell_2\le 3\zeta+2$ such that 
	$\sum_{i=1}^{\ell_1} x_{\pi(i)}=\sum_{i=1}^{\ell_2} x_{\pi(i)}$. Consequently, $\sum_{i=1}^{\ell_2-\ell_1} x_{\pi(i)}=0$. Further notice that by taking out the sequence of integers $x_{\ell_1+1},\cdots,x_{\ell_2}$, the summation is of the remaining integers is still $x\ge 0$

	Hence, as long as $m\ge 6\zeta+2$, we can always select at most $3\zeta+2$ integers whose summation is non-negative, and if we delete them, the summation of the remaining integers is still non-negative. Hence, we can carry on our argument on the remaining integers, and the lemma is proved.
\end{proof}

We can extend the above lemma to higher dimensions using the same basic idea but a much more involved analysis.

In the following we write $\OO^*(x^k)$ to represent a function that is bounded by $(cx)^{k}$ for some constant $c$.
\begin{lemma}\label{lemma:merging-lemma}
	Let $\vex_1,\vex_2,\cdots,\vex_m$ be a sequence of vectors in $\mathbb{Z}^\kappa$ such that $\vex=\sum_{i=1}^m\vex_i$, and $\|\vex_i\|_{\infty}\le \zeta$. Then the $m$ vectors can be partitioned into $m'$ subsets $T_1,T_2,\cdots,T_{m'}$ satisfying that: $\cup_{j=1}^{m'}T_j=\{1,2,\cdots,m\}$, and for every $1\le j\le m'$ it holds that $\sum_{i\in T_j} \vex_i\sqsubseteq \vex$, $|T_j|=\OO^*(\zeta^{\kappa^2})$. 
\end{lemma}
\begin{proof}
	Again we assume without loss of generality that $\vex\ge 0$ (if some of the coordinates are negative, then we take the nagation of every $\vex_i$ and $\vex$ at this coordinate). For $\vex=(\vex^1,\vex^2,\cdots,\vex^{\kappa})$, we further assume that $\vex^1\le \vex^2\le\cdots\le \vex^{\kappa}$ (Notice that here $\vex^j\in\Z$). By Steinitz lemma (Lemma~\ref{lemma:steinitz}), there exists a permutation $\pi$ such that for all $1\le \ell\le m$ it holds that
	\begin{eqnarray*}\label{eq:steinitz1}
		\|\sum_{i=1}^\ell \vex_{\pi(i)}-\frac{\ell-\kappa}{m}\vex\|_{\infty}\le \zeta.
	\end{eqnarray*}
	For simplicity, we reorder all the vectors such that $\vex_{\pi(i)}$ is at the $i$-th location, i.e., we assume that the given vectors $\vex_i$ satisfy that
	\begin{eqnarray}\label{eq:steinitz}
	\|\sum_{i=1}^\ell \vex_{i}-\frac{\ell-\kappa}{m}\vex\|_{\infty}\le \zeta.
	\end{eqnarray}
	
	Our goal is to show the following claim:
	\begin{claim}\label{claim:conformal-sum}
		There always exists a subset $T$ such that, $|T|=\OO^*(\zeta^{\kappa^2}) $, $\sum_{i\in T} \vex_i \sqsubseteq \vex$ and $\vex-\sum_{i\in T} \vex_i\ge 0$.	
	\end{claim}
	If the claim is true, we can iteratively apply it to cut the whole sequence of vectors into subsets $T_1,T_2,\cdots,T_{m'}$ and Lemma~\ref{lemma:merging-lemma} is proved.
	
	To prove Claim~\ref{claim:conformal-sum}, we need the following two claims. For simplicity, we say a subset $T$ is {\em conformal} if $\sum_{i\in T} \vex_i \sqsubseteq \vex$ and $\vex-\sum_{i\in T} \vex_i\ge 0$.  
	
	
	\begin{claim}\label{claim:conformal-sum1}
		For any $1\le j\le \kappa$, if there exists some $\mu_j$ such that 
		$\frac{\mu_j}{m}\vex^j>2\zeta\ge \frac{\mu_j-1}{m}\vex^j$ and $(\mu_j-1)\frac{\vex^j}{\vex^{j-1}}>\kappa+(6\zeta+1)^{j-1}\mu_j$, then there exists a subset $T$ such that $|T|\le 3(6\zeta+1)^{j-1}\mu_j+\kappa$ and $T$ is conformal, i.e., $\sum_{i\in T} \vex_i \sqsubseteq \vex$ and $\vex-\sum_{i\in T} \vex_i\ge 0$.
	\end{claim}
	\begin{proof}[Proof of Claim~\ref{claim:conformal-sum1}] 
		If $m\le 3(6\zeta+1)^{j-1}\mu_j+\kappa$, then we take $T=\{1,2,\cdots,m\}$. In the following we assume that $m>3(6\zeta+1)^{j-1}\mu_j+\kappa$.
		Recall that $\vex^1\le \vex^2\le\cdots\le \vex^{\kappa}$, whereas $\frac{\mu_j}{m}\vex^h>2\zeta$ for any $h\ge j$. Consider an arbitrary subsequence of vectors whose length is $\mu\ge \mu_j$, say, $\vex_{i_0}, \vex_{i_0+1},\cdots,\vex_{i_0+\mu-1}$. By Eq~(\ref{claim:conformal-sum}), we have
		\begin{eqnarray}\label{eq:claim}
		\|\sum_{i=1}^{i_0-1} \vex_{i}-\frac{i_0-1-\kappa}{m}\vex\|_{\infty}\le \zeta, \quad {\text{ and }}\quad  \|\sum_{i=1}^{i_0+\mu-1} \vex_{i}-\frac{i_0+\mu-1-\kappa}{m}\vex\|_{\infty}\le \zeta.
		\end{eqnarray}
		Consequently, for any $h\ge j$, it follows that	
		\begin{eqnarray*}
			\sum_{i=1}^{i_0-1} \vex_{i}^h\le \frac{i_0-1-\kappa}{m}\vex^h+ \zeta, \quad {\text{ and }}\quad  \sum_{i=1}^{i_0+\mu-1} \vex_{i}^h\ge \frac{i_0+\mu-1-\kappa}{m}\vex^h- \zeta.
		\end{eqnarray*}
		Thus,
		\begin{eqnarray}\label{eq:claim1}
		\sum_{i=i_0}^{i_0+\mu-1} \vex_{i}^h\ge \frac{\mu}{m}\vex^h- 2\zeta\ge\frac{\mu_j}{m}\vex^h-2\zeta >0, \quad \forall h\ge j
		\end{eqnarray} 
		This means, the summation of any adjacent $\mu\ge \mu_j$ vectors satisfies that the sum is positive on every $h$-th coordinate for $h\ge j$.
		
		Meanwhile, by Eq~(\ref{eq:claim}) we have
		\begin{eqnarray*}
			\sum_{i=1}^{i_0-1} \vex_{i}^h\ge \frac{i_0-1-\kappa}{m}\vex^h- \zeta, \quad {\text{ and }}\quad  \sum_{i=1}^{i_0+\mu-1} \vex_{i}^h\le \frac{i_0+\mu-1-\kappa}{m}\vex^h+ \zeta.
		\end{eqnarray*}
		Thus, $$\sum_{i=i_0}^{i_0+\mu-1} \vex_{i}^h\le\frac{\mu}{m}\vex^h+2\zeta,\quad \forall h\ge j.$$
		Meanwhile, we have
		$$\sum_{i=1}^{m} \vex_{i}^h\ge\frac{m-\kappa}{m}\vex^h-\zeta\ge \frac{\mu}{m}\vex^h\cdot \frac{m-\kappa}{\mu}-\zeta,\quad \forall h\ge j.$$
		Thus, 
		\begin{eqnarray}\label{eq:claim2}
		\sum_{i=1}^{m} \vex_{i}^h-\sum_{i=i_0}^{i_0+\mu-1} \vex_{i}^h\ge \frac{\mu}{m}\vex^h\cdot (\frac{m-\kappa}{\mu}-1)-3\zeta,\quad \forall h\ge j.
		\end{eqnarray}
		Recall that $\frac{\mu}{m}\vex^h>2\zeta$, as long as $m-\kappa\ge 3\mu$, we know that $\sum_{i=1}^{m} \vex_{i}^h-\sum_{i=i_0}^{i_0+\mu-1} \vex_{i}^h> 0$ for all $h\ge j$.

		Next we consider the $h$-th coordinate for $h<j$. Recall that $\frac{\mu_j-1}{m}\vex^j\le 2\zeta$. As $(\mu_j-1)\frac{\vex^j}{\vex^{j-1}}>\kappa+(6\zeta+1)^{j-1}\mu_j$, it follows directly that
		$$\frac{\kappa+(6\zeta+1)^{j-1}\mu_j}{m}\vex^h\le 2\zeta, \quad \forall h\le j-1.$$ Now we consider the $1+(6\zeta+1)^{j-1}$ vectors $\sum_{i=1}^{\ell}\vex_i$ for $\ell\in \mathcal{L}_j$ where $\mathcal{L}_j=\{\kappa,\kappa+\mu_j,\kappa+2\mu_j,\cdots,\kappa+(6\zeta+1)^{j-1}\mu_j\}$. For any $\ell\in \mathcal{L}_j$, it is clear that
		$$|\sum_{i=1}^{\ell}x_i^h|\le \frac{\ell-\kappa}{m}\vex^h+2\zeta\le 3\zeta, \quad \forall h\le j-1$$
		that is, $\sum_{i=1}^{\ell}x_i^h\in\{-3\zeta,-3\zeta+1,\cdots,3\zeta\}$ for all $\ell\in \mathcal{L}_j$ and $h\le j-1$. Hence, if we consider the projection of $\sum_{i=1}^\ell\vex_i$ onto its first $j-1$ coordinates, this projection lies within $\{-3\zeta,-3\zeta+1,\cdots,3\zeta\}^{j-1}$, implying that there must exist $\ell_1<\ell_2$ such that $\sum_{i=1}^{\ell_1}\vex_i$ and $\sum_{i=1}^{\ell_2}\vex_i$ have the same projection. Hence, the first $j-1$ coordinates of $\sum_{i=\ell_1+1}^{\ell_2}\vex_i$ are all $0$. Furthermore, we observe the followings: 1). $\ell_2-\ell_1>\mu_j'$, whereas for $h\ge j$, the $h$-th coordinate of $\sum_{i=\ell_1+1}^{\ell_2}\vex_i$ is positive according to Eq~(\ref{eq:claim1}). 2). $\ell_2-\ell_1\le (6\zeta+1)^{j-1}\mu_j$ and $m\ge 3(6\zeta+1)^{j-1}\mu_j+\kappa$, whereas for $h\ge j$, the $h$-th coordinate of $\sum_{i=1}^m\vex_i-\sum_{i=\ell_1+1}^{\ell_2}\vex_i$ is also positive according to Eq~(\ref{eq:claim2}). Hence, $\sum_{i=\ell_1+1}^{\ell_2}\vex_i\sqsubseteq \vex$ and $\sum_{i=1}^m\vex_i-\sum_{i=\ell_1+1}^{\ell_2}\vex_i\ge 0$, i,e, by taking $T=\{\ell_1+1,\ell_1+2,\cdots,\ell_2\}$, Claim~\ref{claim:conformal-sum1} is true. \end{proof}

	Now we come to the proof of Claim~\ref{claim:conformal-sum}.
	
	\begin{proof}[Proof of Claim~\ref{claim:conformal-sum}]
		We prove the claim by induction on the following hypothesis.
		
		\smallskip
		\noindent\textbf{Statement (Hypothesis):}
		For $1\le j\le \kappa$, either there exists some $T$ which is conformal (i.e., $\sum_{i\in T} \vex_i \sqsubseteq \vex$ and $\vex-\sum_{i\in T} \vex_i\ge 0$) and $|T|=\OO^*(\zeta^{(\kappa-j+1)\kappa})$, or there exists some $\mu_{j}=\OO^*(\zeta^{(\kappa-j+1)\kappa})$ such that $\frac{\mu_{j}}{m}\vex^{j}> 2\zeta \ge \frac{\mu_{j}-1}{m}\vex^{j}$. 
		
		\smallskip
		We first prove the statement for $j=k$.
		Taking $\mu_{\kappa}'=(6\zeta+1)^{\kappa}+\kappa=\OO^*(\zeta^{\kappa})$. There are two possibilities. 
		
		If $\frac{\mu_{\kappa}'}{m}\vex^{k}\le 2\zeta $, then $\frac{\mu_{\kappa}'}{m}\vex^{j}\le 2\zeta $ for all $1\le j\le k$. Consequently, for $\ell\in \mathcal{L}=\{\kappa,\kappa+1,\cdots,\mu_{\kappa}'\}$, we have
		$$\|\sum_{i=1}^{\ell}\vex_i\|_{\infty}\le \frac{\mu_{\kappa}'}{m}\vex^{k}+\zeta\le 3\zeta, \quad \forall i\in \mathcal{L}$$
		i.e., $\sum_{i=1}^{\ell}\vex_i\in\{-3\zeta,-3\zeta+1,\cdots,3\zeta\}^{\kappa}$.  Since $|\mathcal{L}|=(6\zeta+1)^{\kappa}+1$, there exist $\ell_1<\ell_2$ and $\ell_1,\ell_2\in\mathcal{L}$ such that $\sum_{i=1}^{\ell_1}\vex_i=\sum_{i=1}^{\ell_2}\vex_i$, i.e., $\sum_{i=\ell_1+1}^{\ell_2}\vex_i=0$. Taking $T=\{\ell_1+1,\cdots,\ell_2\}$, the statement is true.
		
		Otherwise, $\frac{\mu_{\kappa}'}{m}\vex^{k}> 2\zeta $. Then there exists some $\mu_{\kappa}\le \mu_{\kappa}'=\OO^*(\zeta^{k})$ such that $\frac{\mu_{\kappa}}{m}\vex^{k}> 2\zeta \ge \frac{\mu_{\kappa}-1}{m}\vex^{k}$. That is, the statement is also true.
		
		Suppose the statement (hypothesis) holds for $j\ge j_0$, we prove it for $j=j_0-1$. According to the statement, either there exists some $T$ satisfying Claim~\ref{claim:conformal-sum} with $|T|=\OO^*(\zeta^{(\kappa-j_0+1)\kappa})$, or there exists some $\mu_{j_0}=\OO^*(\zeta^{(\kappa-j_0+1)\kappa})$ such that $\frac{\mu_{j_0}}{m}\vex^{j_0}> 2\zeta \ge \frac{\mu_{j_0}-1}{m}\vex^{j_0}$. If the former case is true, then obviously the same $T$ satisfies that $|T|\le \OO^*(\zeta^{(\kappa-j_0+2)\kappa})$, implying that the statement is true for $j=j_0-1$. Hence, from now on we assume the latter case is true. According to Claim~\ref{claim:conformal-sum1}, if $(\mu_{j_0}-1)\frac{\vex^{j_0}}{\vex^{j_0-1}}>\kappa+(6\zeta+1)^{j_0-1}\mu_{j_0}$, then there exists a subset $T$ which is conformal and $|T|\le (6\zeta+1)^{j_0-1}\mu_{j_0}$. Plugging in $\mu_{j_0}=\OO^*(\zeta^{(\kappa-j_0+1)\kappa})$, we have $|T|=\OO^*(\zeta^{\kappa-j_0+2})$, that is, if $(\mu_{j_0}-1)\frac{\vex^{j_0}}{\vex^{j_0-1}}>\kappa+(6\zeta+1)^{j_0-1}\mu_{j_0}$, then the statement holds for $j=j_0-1$. Thus, in the following we assume that $(\mu_{j_0}-1)\frac{\vex^{j_0}}{\vex^{j_0-1}}\le \kappa+(6\zeta+1)^{j_0-1}\mu_{j_0}$. Notice that $\vex^j/m\le \zeta$ (as $\|\vex_i\|_{\infty}\le \zeta$). According to $\frac{\mu_{j_0}}{m}\vex^{j_0}> 2\zeta$, we know $\mu_{j_0}\ge 2$, whereas
		$$\frac{\vex^{j_0}}{\vex^{j_0-1}}\le \frac{\kappa+(6\zeta+1)^{j_0-1}\mu_{j_0}}{\mu_{j_0}-1},$$
		and consequently
		$$\frac{\mu_{j_0}}{\mu_{j_0}-1}\cdot\frac{\kappa+(6\zeta+1)^{j_0-1}\mu_{j_0}}{m}\vex^{j_0-1}> 2\zeta.$$
		Since $\mu_{j_0}=\OO^*(\zeta^{(\kappa-j_0+1)\kappa})$, we can conclude that $\frac{\mu_{j_0}}{\mu_{j_0}-1}\cdot[{\kappa+(6\zeta+1)^{j_0-1}\mu_{j_0}}]=\OO^*(\zeta^{(\kappa-j_0+2)\kappa})$, hence, there exists some $\mu_{j_0-1}=\OO^*(\zeta^{(\kappa-j_0+2)\kappa})$ such that $\frac{\mu_{j_0-1}}{m}\vex^{j_0-1}>2\zeta\ge \frac{\mu_{j_0-1}-1}{m}\vex^{j_0-1}$. Thus, the statement holds for $j=j_0-1$.
		
		Now we have proved the statement for all $1\le j\le \kappa$. Taking $j=1$, either there exists some subset $T$ which is conformal and $|T|= \OO^*(\zeta^{\kappa^2})$, and Claim~\ref{claim:conformal-sum} is proved; Or there exists some $\mu_1=\OO^*(\zeta^{\kappa^2})$ such that $\frac{\mu_1}{m}\vex^1>2\zeta$. As $\vex^1\le \vex^j$ for all $j\le \kappa$, it holds that $\frac{\mu_1}{m}\vex^j>2\zeta$. There are two possibilities. If $m\le 2\mu_1+\kappa=\OO^*(\zeta^{\kappa^2})$, we simply take $T=\{1,2,\cdots,m\}$. Otherwise, we have 	
		\begin{eqnarray*}
			\|\sum_{i=1}^{\mu_1+\kappa} \vex_{i}-\frac{\mu_1}{m}\vex\|_{\infty}\le \zeta, \quad \textrm{ and }\quad  \|\sum_{i=1}^{m} \vex_{i}-\frac{m-\kappa}{m}\vex\|_{\infty}\le \zeta.
		\end{eqnarray*}
		Consequently, for any $1\le j\le \kappa$, it follows that	
		\begin{eqnarray*}
			0\le \frac{\mu_1}{m}\vex^j- \zeta	\le \sum_{i=1}^{\mu_1+\kappa} \vex_{i}^j\le \frac{\mu_1}{m}\vex^j+ \zeta, \quad {\text{ and }}\quad  \sum_{i=1}^{m} \vex_{i}^j\ge \frac{m-\kappa}{m}\vex^h- \zeta\ge 2\frac{\mu_1}{m}\vex^h- \zeta>\frac{\mu_1}{m}\vex^h+ \zeta.
		\end{eqnarray*}
		Hence, taking $T=\{1,2,\cdots,\mu_1+\kappa\}$ we have that $\sum_{i\in T}\vex_i\sqsubseteq \vex$ and $\sum_{i=1}^m\vex_i-\sum_{i\in T}\vex_i\ge 0$, and $|T|=\OO^*(\zeta^{\kappa^2})$, i.e., Claim~\ref{claim:conformal-sum} is proved.
	\end{proof}
	Iterratively applying Claim~\ref{claim:conformal-sum}, Lemma~\ref{lemma:merging-lemma} is proved. 
\end{proof}
\noindent\textbf{Remark.} It is notable that a weaker version of Lemma~\ref{lemma:merging-lemma} can also be proved, in a much simpler way, by iteratively applying Lemma~\ref{lemma:merging-lemma-1} to each individual dimension. However, by doing so we get an upper bound of $\OO^*(\zeta^{2^{\kappa}})$ on $|T_j|$'s, which is much worse.

\subsubsection{A decomposition in two orthants}\label{subsec:decompose-two}

The goal of this subsection is to provide a semi-sign-compatible decomposition. More precisely, we prove the following:

\begin{lemma}\label{lemma:decompose-bounded}
	For any $\vey\in \ker_{\Z}(H_0)$, there exist $q=\OFPT(1)$ vectors $\vece_h$, $\alpha_h,\beta_{\ell}\in\Z_+$ and at most $2nt_A-1$ vectors $\ved_\ell=(0,\veg_\ell(E))$ where $\veg_\ell(E)\in\G(E)$ such that $\vece_h^0\sqsubseteq \vey^0$, $\|\vece_h\|_{\infty}\le \xi'=\OFPT(1)$,
	$\vey=\sum_{h=1}^q\alpha_h\vece_h+\sum_{\ell}\beta_{\ell}\ved_\ell$, and moreover, all the $\vece_h$'s lie in the same orthant, and all the $\ved_\ell$'s lie in the same orthant. 	
\end{lemma} 
Note that $\vece_h$ and $\ved_\ell$ do not necessarily lie in the same orthant. 

\begin{proof}
		According to Theorem~\ref{lemma:3-infty-bound}, there exist $\vece_1,\vece_2,\cdots,\vece_k$ with $\|\vece_h\|_{\infty}\le \xi$, $\vece_h^0\sqsubseteq \vey^0$ such that $\vey=\sum_{h=1}^k \vece_h$. Let $\veu_1,\veu_2,\cdots,\veu_{\eta}$ be all the $t_B$-dimensional vectors whose $\ell_{\infty}$-norm is bounded by $\xi$, then $\eta=O(\xi^{t_B})=\OFPT(1)$. For any $\veu_j$, we pick an arbitrary $\bar{\vece}_j$ such that $H_0\bar{\vece}_j=0$ and $\bar{\vece}_j^0=\veu_j$. Note that such a $\bar{\vece}_j$ can be found out by solving an $n$-fold IP, which can be done in $\OFPT(n^3L)$ time~\cite{hemmecke2013n}. Among $\vece_1$ to $\vece_k$, we define $K_j=\{\vece_h:\vece_h^0=\veu_j, 1\le h\le k\}$. We have
	$$\vey=\sum_{j=1}^\eta \bar{\vece}_j\cdot |K_j|+\sum_{j=1}^\eta\sum_{\vece_h\in K_j}(\vece_h-\bar{\vece}_j)$$
	
	Since $H_0\vece_h=0$, it is easy to see that $\sum_{j=1}^q\sum_{\vece_h\in K_j}(\vece_h-\bar{\vece}_j)=(0,\vece')$ where $\vece'$ is a feasible solution to $E\vex=0$ where $E$ is an $n$-fold matrix. According to \cite{hemmecke2013n}, there exists at most $2nt_A-1$ vectors $\veg_\ell(E)\in\G(E)$, $\veg_\ell(E)\sqsubseteq \vece'$ and $\beta_\ell\in\Z_+$ such that $\vece'=\sum_{\ell}\beta_\ell\veg_\ell(E)$. Define $\ved_\ell=(0,\veg_\ell(E))$, we have
	$$\sum_{j=1}^\eta\sum_{\vece_h\in K_j}(\vece_h-\bar{\vece}_j)=\sum_{\ell}\beta_\ell\ved_\ell.$$
	
Now we consider the $\bar{\vece}_j$'s where $K_j\neq\emptyset$. If they are all sign-compatible, the lemma is proved. Otherwise we try to apply Lemma~\ref{lemma:merging-lemma} to the sequence of vectors that consists of $|K_j|$ copies of $\bar{\vece}_j$. Note that we cannot directly apply the lemma as $\bar{\vece}_j$'s have very high dimensions. However, if we consider the bricks $\bar{\vece}_j^i$, since $\|\bar{\vece}_j^i\|_{\infty}\le \xi$, there are at most $\xi^{O(t_A)}$ different kinds of bricks. Consequently, if we consider the vectors that consists of the $\eta$ bricks $(\bar{\vece}_1^i,\bar{\vece}_2^i,\cdots,\bar{\vece}_\eta^i)$, there are at most $\theta=\xi^{O(\eta t_A)}=\OFPT(1)$ different kinds of vectors for $1\le i\le n$. We let these vectors be $\phi_1,\phi_2,\cdots,\phi_\theta$. Now we consider the \lq\lq reduced\rq\rq\, vectors $Rd(\bar{\vece}_j)$ that only consists of distinct vectors. More precisely, we define the set of indices $In_j=\{i:(\bar{\vece}_1^i,\bar{\vece}_2^i,\cdots,\bar{\vece}_\eta^i)=\phi_j\}$. For each $In_j$, we pick an arbitrary $i_j\in In_j$ and define a $(t_B+\theta t_A)$-dimensional vector $Rd(\bar{\vece}_h)=(\bar{\vece}_h^0,\bar{\vece}_h^{i_1}, \bar{\vece}_h^{i_2},\cdots,\bar{\vece}_h^{i_\theta})$. Note that $\bar{\vece}_h$ is simply a vector that copies the bricks of $Rd(\bar{\vece}_h)$ for multiple times. Now we consider the sequence that consists of $|K_j|$ copies of $Rd(\bar{\vece}_j)$. Applying Lemma~\ref{lemma:merging-lemma}, we can divide these vectors into disjoint subsets $S_1,S_2,\cdots, S_m$ such that the summation of vectors in each subset is sign-compatible, and each subset has cardinality bounded by $\xi^{O(\theta t_A+t_B)}$. Consequently, $\sum_{h\in S_j}\bar{\vece}_h$'s are also sign-compatible. Let $\vece_j=\sum_{h\in S_j}\bar{\vece}_h$, we have
	$$\vey=\sum_{j=1}^m {\vece}_j+\sum_{\ell}\beta_\ell\ved_\ell.$$
	It remains to show there are at most $\OFPT(1)$ different kinds of $\vece_j$'s. To see this, consider the reduced vector $\vece_j'=\sum_{i\in S_j}\hat{\vece}_i$ and notice that $\vece_j$ is duplicating the bricks of $\vece_j'$ at locations indicated by $In_k$. Hence, it suffices to show that there are at most $\OFPT(1)$ different kinds of $\vece_j'$'s. Note that $\|\vece_j'\|_{\infty}\le \xi\cdot  \xi^{O(\theta t_A+t_B)}$, and it is of $(t_B+\theta t_A)$-dimensional. Hence, there are only $\OFPT(1)$ different kinds of $\vece_j'$'s, and the lemma is proved.
\end{proof}

Our next goal is to further make $\vece_h$'s and $\ved_h$'s sign-compatible. Given $\vey$ and a decomposition satisfying Lemma~\ref{lemma:decompose-bounded}, we call $\vece_h$'s as the principle vectors and $\ved_h$'s as the add-ons. The basic idea is to merge principle vectors and add-ons such that they become sign-compatible, and we will mainly use Lemma~\ref{lemma:merging-lemma} to achieve this. However, there is a problem in applying Lemma~\ref{lemma:merging-lemma} directly as the dimension is too high. Again we try to utilize the idea in the proof of Lemma~\ref{lemma:decompose-bounded}: note that principle vectors are good in the sense that they can be reduced to lower dimensional vectors such that they are duplicating the bricks of lower dimensional vectors in fixed locations. While add-ons do not have such a nice structure, they are \lq\lq sparse\rq\rq\, according to Lemma~\ref{lemma:cite-nfold}, that is, only an $\OFPT(1)$ number of their bricks are non-zero. This will allow us to achieve the desired merging.



\subsubsection{Defining types of bricks}\label{subsec:type}
Prior to our merging process, let $\Gamma$ be some positive integer to be determined later. We will eventually set its value within $\OFPT(1)$, but for ease of analysis on its value at the end, we will first treat it as an unbounded parameter and write $\OFPT(\Gamma)$ in the following. 

We first define a {\em quantity type}. For every $t_A$-dimensional vector $\vex=(x_1,x_2,\cdots,x_{t_A})$, we compare every coordinate $x_j$ with $\Gamma$. If $x_j\le \Gamma$, we say the $j$-th coordinate of $\vex$ is {\em small}. Otherwise, we say it is {\em large}. A large coordinate may be positive or negative, hence each coordinate of $\vex$ can be of three kinds: small, positive large and negative large. 
The {\em quantity type} of each $\vex$ is defined as a $t_B$-dimensional vector which stores the kind of each $x_j$. It is obvious there are at most $3^{t_A}$ different quantity types.

Next, we define a {\em principle type} for every $\vey^i$. Note that $\|{\vece}_j\|_{\infty}\le \xi'$. For each $1\le i\le n$, we define the vector $({\vece}_1^i,{\vece}_2^i,\cdots,{\vece}_q^i)$ as the {\em principle type} of $\vey^i$. There are at most $(\xi')^{O(qt_A)}=\OFPT(1)$ different principle types.


Consider the bricks of $\vey$. $\vey^i$'s with the same quantity type and principle type are called to have the same {\em type}. There are at most $6^{t_A}\cdot (\xi')^{O(qt_A)}=\OFPT(1)$ different types. We pick an arbitrary brick, say, brick $1$ as a specific brick and let $N_1=\{1\}$. For the remaining bricks (brick $2$ to brick $n$), we divide them into $\sigma-1=\OFPT(1)$ sets such that bricks in the same set have the same type, i.e., we let $N_2,\cdots,N_\sigma$ be the set of indices of the bricks that have the same type, and let $n_j=|N_j|$. For simplicity, we reorder the bricks of $\vey$ such that $N_j=\{\iota_{j-1}+1,\iota_{j-1}+2,\cdots,\iota_{j-1}+n_j\}$ where $\iota_{j-1}=n_1+n_2+\cdots+n_{j-1}$. 

\subsubsection{Equalization}\label{subsec:centralize}
According to Lemma~\ref{lemma:cite-nfold}, every $\ved_h^i$, as well as $\sum_i \ved_h^i$, is the summation of at most $\OFPT(1)$ elements of $\G(A)$. Let $\vev_1,\vev_2,\cdots,\vev_\lambda$ be all the non-zero $t_A$-dimensional vectors that $\ved_h^i$ can take. For simplicity, we allow $\ved_h$'s to be the same and rewrite Eq~(\ref{eq:decompose-main}) as
\begin{eqnarray*}\label{eq:decompose-main-1}
	\vey=\sum_{h=1}^{q}\alpha_h{\vece}_h+\sum_{h}\ved_h.
\end{eqnarray*}
Note that in the above summation we simply add each $\ved_h$ separately by $\beta_h$ times.

For ease of description, let us now take a scheduling point of view. We assume there are $n$ machines. The $t_B$-dimensional load of machine $i$ is defined by $\vey^i$. This load is contributed by two parts, $\sum_{h=1}^q\alpha_h{\vece}_h^i$ and $\sum_{h}\vemd_h^i$. For every $i\in N_j$, the first part $\sum_{h=1}^q\alpha_h{\vece}_h^i$ is the same, while the second part might be different.
We can view each $\vev_k$ as a $t_A$-dimensional job. Obviously there are only $\lambda=\OFPT(1)$ different kinds of jobs, albeit that each job may have multiple identical copies. Let $\psi(j,k)$ be the total number of copies of job $k$ on machines in $N_j$.

We define a vector $\vey_f$ such that $\vey^0_f=\vey^0$, $\vey^1_f=\vey^1$ and
\begin{eqnarray}\label{eq:average}
\vey_f^k= \frac{1}{n_j}\cdot \sum_{i\in N_j}\vey^i=\sum_{h=1}^q\alpha_h\vece_h^i+\frac{1}{|N_j|}\cdot\sum_{i\in N_j}\sum_{h}\ved_h^i, \quad \forall k\in N_j, 2\le j\le \sigma
\end{eqnarray}

Ideally, we would like to argue on $\vey_f$. However, $\vey_f$ may be fractional. Therefore, in the following we define an integral vector $\tilde{\vey}\approx \vey_f$ and call it the {\em equalization} of $\vey$.

We give the precise definition of $\tilde{\vey}$ as follows. Let $\psi(j,k)$ be the number of copies of job $\vev_k$ on machines in $N_j$. We (almost) evenly distribute these jobs among machines such that every machine gets $\lfloor\psi(j,k)/n_j\rfloor$ or $\lceil \psi(j,k)/n_j\rceil$ copies. To make it unique, we further restrict that machines with smaller indices in $N_j$ always have the same or more number of copies. By doing so, we construct a new vector $\tilde{\vey}$. Note that $\tilde{\vey}$ consists of the same number of jobs as $\vey$, only that jobs are distributed among machines in a different (more even) way. 

As we re-distribute $\vev_k$'s such that for every machine in $N_j$, the number of copies of each $\vev_k$ differs by at most $1$, the following lemma is straightforward.

\begin{lemma}\label{lemma:average}
	$$\|\tilde{\vey}^i-\vey_f^i\|_{\infty}\le \sum_{k=1}^{\lambda} \|\vev_k\|_{\infty}.$$
\end{lemma}

\subsubsection{Decomposition of $\tilde{\vey}$}\label{subsec:decompose-tilde-vey}
We create new $(t_B+nt_A)$-dimensional vectors in the following way. For simplicity, we define $\psi^q(j,k)=\lfloor\psi(j,k)/n_j\rfloor$ and $\psi^r(j,k)=\psi(j,k)-n_j\psi^q(j,k)$, i.e., they are the quotient and residue of $\psi(j,k)$ divided by $n_j$, respectively. 
For every $2\le j\le \sigma$, we create $\psi^q(j,k)$ copies of a vector $\vemd(j,k)$ and one copy of $\overline{\vemd}(j,k)$ such that
\[
\vemd^i(j,k)=\left\{
\begin{array}{ll}
\vev_k,\hspace{9mm} i\in N_j\\
-n_j\vev_k, \quad i=1\\
0, \quad otherwise
\end{array}
\right. \hspace{10mm}
\overline{\vemd}^i(j,k)=\left\{
\begin{array}{ll}
\vev_k,\hspace{20mm} \iota_j\le i\le \iota_j+\psi^r(j,k)\\
-\psi^r(j,k)\cdot\vev_k, \hspace{25mm} i=1\\
0, \hspace{37mm} otherwise
\end{array}
\right.
\]

Using the above notations, it is clear that for any $i\ge 2$, $\tilde{\vey}^i$ consists of $\psi^q(j,k)$ copies of $\vemd^i(j,k)$ and one copy of $\overline{\vemd}^i(j,k)$, i.e., we have the following:
\begin{eqnarray}\label{eq:decompose-extra}
\tilde{\vey}^i=\sum_{h=1}^q\alpha_h\vece_h^i+\sum_{j=1}^{\sigma}\sum_{k=1}^{\lambda}\left(\psi^q(j,k)\cdot\vemd^i(j,k)+\overline{\vemd}^i(j,k) \right), \quad \forall i\ge 2
\end{eqnarray}
The above equation is also true for $i=0$ as $\vemd^0(j,k)=\overline{\vemd}^0(j,k)=0$ for every $1\le j\le \sigma$, $1\le k\le \lambda$.

Furthermore, we have the following observations which follow directly from the definitions of $\vemd(j,k)$ and $\overline{\vemd}(j,k)$.

\begin{observation}
	$H_0\cdot \vemd(j,k)=0$ and $H_0\cdot \overline{\vemd}(j,k)=0$ for all $1\le j\le \sigma$ and $1\le k\le \lambda$.
\end{observation}

\begin{observation}\label{obs:size}
	For any $i=0$ or $2\le i\le n$, $\vemd^i(j,k)=\OFPT(1),\overline{\vemd}^i(j,k)=\OFPT(1)$; For $i=1$, $\vemd^1(j,k)=\OFPT(n),\overline{\vemd}^1(j,k)=\OFPT(n)$.
\end{observation}

It is clear that Eq~(\ref{eq:decompose-extra}) does not necessarily hold for $i=1$. Let us consider 
$$\eta=\tilde{\vey}^1-\sum_{i=1}^q\alpha_h\vece_h^1-\sum_{j=1}^{\sigma}\sum_{k=1}^{\lambda}\left(\psi^q(j,k)\cdot\vemd^1(j,k)+\overline{\vemd}^1(j,k) \right).$$
We have the following lemma.
\begin{lemma}\label{lemma:machine-1}
	$D\eta=0$.
\end{lemma}
\begin{proof}
	Note that $H_0\ved_{\ell}=0$ for each $\ved_{\ell}$, whereas $D\sum_{i=1}^n \sum_{h}\ved_h^i=0$. Since $\tilde{\vey}$ is constructed from $\vey$ by re-distributing the bricks $\ved_h^i$ (i.e., by shifting it from brick $i$ to brick $i'$), it holds that
	$$D\sum_{i=1}^n(\tilde{\vey}^i-\sum_{h=1}^q\vece_h^i)=0.$$
	Plugging in Eq~(\ref{eq:decompose-extra}), we have
	\begin{eqnarray*}
		0&=&D\tilde{\vey}^1+D\sum_{i=2}^n\tilde{\vey}^i-D\sum_{i=1}^n\sum_{h=1}^q\vece_h^i\\
		&=&D\tilde{\vey}^1+D\sum_{i=2}^n\left(\sum_{h=1}^q\alpha_h\vece_h^i+\sum_{j=1}^{\sigma}\sum_{k=1}^{\lambda}\left(\psi^q(j,k)\cdot\vemd^i(j,k)+\overline{\vemd}^i(j,k) \right) \right)-D\sum_{i=1}^n\sum_{h=1}^q\vece_h^i\\
		&=&D\tilde{\vey}^1-D\sum_{h=1}^q\vece_h^1+D\sum_{j=1}^{\sigma}\sum_{k=1}^{\lambda}\left(-\psi^q(j,k)\cdot\vemd^1(j,k)-\overline{\vemd}^1(j,k) \right)=D\eta.	
	\end{eqnarray*} 
	Here the third equation makes use of the fact that $\vemd^1(j,k)=-\sum_{i=2}^n\vemd^i(j,k)$ and $\overline{\vemd}^1(j,k)=-\sum_{i=2}^n\overline{\vemd}^i(j,k)$.	
\end{proof}
Recall that by definition $\vey^1-\sum_{h=1}^q\vece_h^1$ is a weighted sum of $\vev_k$'s, and so is $\sum_{j=1}^{\sigma}\sum_{k=1}^{\lambda}\left(\psi^q(j,k)\cdot\vemd^1(j,k)+\overline{\vemd}^1(j,k) \right)$. Hence, $\eta$ is also a weighted sum of $\vev_k$'s, and we let $\eta=\sum_{k}\gamma_k\vev_k$ where $\gamma_k\in \Z$ for $1\le k\le \lambda$. According to Lemma~\ref{lemma:machine-1}, we have that 
$$\sum_{k=1}^{\lambda}\gamma_k \cdot D\vev_k=0.$$
Equivalently, the above equation can be written as
$$(\gamma_1,\gamma_2,\cdots,\gamma_{\lambda})\cdot [D\vev_1,D\vev_2,\cdots,D\vev_{\lambda}]=0.$$
Consequently, if we define the matrix $DV=[D\vev_1,D\vev_2,\cdots,D\vev_{\lambda}]$, which is an $\OFPT(1)\times \OFPT(1)$ matrix, then there exist $\gamma_h'\in\Z_+$ and $\veg_h(DV)\in \G(DV)$, $\veg_h(DV)\sqsubseteq (\gamma_1,\gamma_2,\cdots,\gamma_k)$ such that
$$(\gamma_1,\gamma_2,\cdots,\gamma_{\lambda})=\sum_{h=1}^\omega \gamma_h'\veg_h(DV).$$
where $\omega\le |\G(DV)|=\OFPT(1)$. Consequently, we have
$$\eta=\sum_{h=1}^{\omega} \gamma_h'\left(\sum_{k=1}^{\lambda} \veg_h^k(DV)\cdot \vev_k\right),$$
Note that here each $\veg_h^k(DV)\in \Z$ is the $k$-th coordinate of $\veg_h(DV)$. Furthermore, $\|\veg_h(DV)\|_{\infty}=\OFPT(1)$.

We define new vectors $\veod(h)$ such that
\[
\veod^i(h)=\left\{
\begin{array}{ll}
\sum_{k=1}^\lambda \veg_h^k(DV)\vev_k,\hspace{9mm} i=1\\
0, \hspace{23mm} otherwise
\end{array}
\right. 
\]
Recall that $A\vev_k=0$ and $\sum_{k=1}^\lambda  \veg_h^k(DV)\cdot D\vev_k=0$, we have the following observation.

\begin{observation}
	$\|\veod(h)\|_{\infty}=\OFPT(1)$ and $H_0\cdot \veod(h)=0$ for all $1\le h\le \omega$.
\end{observation}

Now we derive the following decomposition of $\tilde{\vey}$:
\begin{eqnarray}\label{eq:decompose}
\tilde{\vey}=\sum_{h=1}^q\alpha_h\vece_h+\sum_{j=1}^{\sigma}\sum_{k=1}^{\lambda}\left(\psi^q(j,k)\cdot\vemd(j,k)+\overline{\vemd}(j,k) \right)+\sum_{h=1}^\omega \gamma_h'\cdot \veod(h),
\end{eqnarray}

\subsubsection{A sign-compatible decomposition of $\tilde{\vey}$}\label{subsec:decompose-tildevey-compatible}

We will find a sign-compatible decomposition of $\tilde{y}$ in this subsection, and show in the next subsection that at least one element of the decomposition lie in the same orthant of $\bar{\vey}$. 

Recall Eq~(\ref{eq:decompose}). We observe that all the vectors involved have a nice structure in the sense that they can be divided into $\OFPT(1)$ segments where every segment consists of identical bricks. More precisely, for every $1\le j\le \sigma$, let $\pi_j$ be the permutation of $\{1,2,\cdots,k\}$ such that the $\lambda$ residues can be ordered as $\psi^r(j,\pi_j(1))\le \psi^r(j,\pi_j(2))\le\cdots\le \psi^r(j,\pi_j(\lambda))$. Additionally, we define $\psi^r(j,\pi_j(0))=\iota_{j-1}+1$ for $j\ge 2$, $\psi^r(j,\pi_j(\lambda+1))=\iota_j-1$ and $\psi^r(1,\pi_j(0))=2$ (as machine $1$ is special and should be excluded). 
We can divide the $n+1$ bricks of a $(t_B+nt_A)$-dimensional vector into $2+2(\lambda+1)\sigma$ groups as follows: 
\begin{itemize}
	\item Group 0 consists of brick 0 (the first $t_B$ dimensions), which is $0$ for all the $\vemd(j,k)$ and $\overline{\vemd}(j,k)$.
	\item Group 1 consists of only machine (brick) 1.
	\item For $1\le j\le \sigma$ and $1\le k\le \lambda+1$, Group $k+1+(j-1)(\lambda+1)$ consists of brick $\psi^r(j,\pi(k-1))+1$ to brick $\psi^r(j,\pi(k))$.
\end{itemize}
Hence, each vector is divided into $2+2(\lambda+1)\sigma$ segments where each segments contains its bricks within one group. See the following figure as an illustration of the grouping. Here circles of different colors represent different $\vev_k$'s. Note that if we take a \lq\lq snapshot\rq\rq\, of any vector ($\vece_h$ or $\vemd(j,k)$) on the bricks within a group (see the bricks among two ajacent red lines in the figure), we see that all of these bricks are identical (for otherwise some of the residues shall lie within the indices of these bricks, which contradicts the grouping). More precisely, we have the following.

\begin{figure}\label{fig:1}
	\begin{center}
		\includegraphics[scale=0.4]{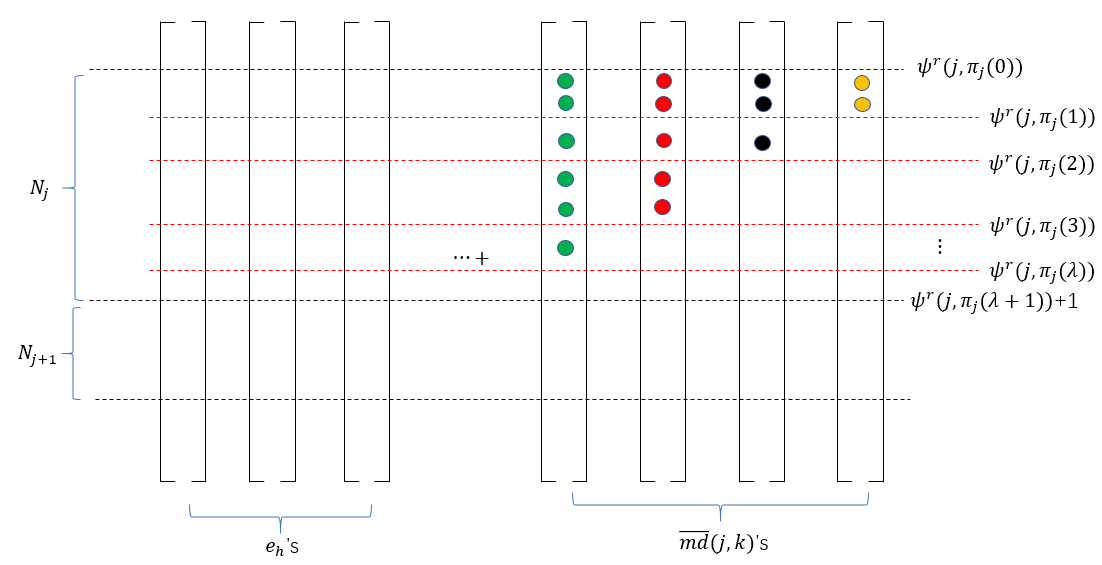}
	\end{center}	
\end{figure}

\begin{observation}\label{obs:decompose}
	Let $Gr_\ell$ be the indices of bricks in Group $\ell$, then for every $i_1,i_2\in Gr_\ell$, we have $\vece_h^{i_1}=\vece_h^{i_2}$ and $\vemd^{i_1}(j,k)=\vemd^{i_2}(j,k)$ for $1\le h\le q$, $1\le j\le \sigma$, $1\le k\le \lambda$.
\end{observation}
Furthermore, notice that $Gr_\ell$'s is a further sub-division of $N_1,N_2,\cdots,N_{\sigma}$, hence we have the following observation.

\begin{observation}
	For any $i_1,i_2\in Gr_\ell$, $\vey^{i_1}$ and $\vey^{i_2}$ have the same type.
\end{observation}

Now we are able to define reduced vectors. For $\vez=\vece_h$ or $\vemd(j,k)$ or $\overline{\vemd}(j,k)$ or $\veod(h)$ or $\tilde{\vey}$, we define $Rd(\vez)$ as a $(t_B+t_A+2(\lambda+1)\sigma t_A)$-dimensional vector that consists of $2+2(\lambda+1)\sigma$ bricks where the $\ell$-th brick $Rd^\ell(\vez)$ equals any brick of $\vez$ in the group $Gr_\ell$ (as they are the same by Observation~\ref{obs:decompose}). Furthermore, Eq~(\ref{eq:decompose}) implies the following: 
\begin{eqnarray}\label{eq:decompose2}
Rd(\tilde{\vey})=\sum_{i=1}^q\alpha_h Rd(\vece_h)+\sum_{j=1}^{\sigma}\sum_{k=1}^{\lambda}\left(\psi^q(j,k)\cdot Rd(\vemd(j,k))+Rd(\overline{\vemd}(j,k))\right)+\sum_{h=1}^\omega \gamma_h'\cdot Rd(\veod(h)) .
\end{eqnarray}

If we want to make the rightside of Eq~(\ref{eq:decompose}) into a sign-compatible summation, it suffices to make the above Eq~(\ref{eq:decompose2}) into a sign-compatible summation, and this is achievable by utilizing Lemma~\ref{lemma:merging-lemma}. To derive a good bound, we will apply Lemma~\ref{lemma:merging-lemma} twice in a separate way. 

By Observation~\ref{obs:size}, we have the following.
\begin{observation}\label{obs:size-2}
	For $i=0$ or $i\ge 2$, $Rd^i(\vemd(j,k)),Rd^i(\overline{\vemd}(j,k))=\OFPT(1)$; For $i=1$, $Rd^1(\vemd(j,k)),Rd^1(\overline{\vemd}(j,k))=\OFPT(n)$.
\end{observation}

We now view the rightside of Eq~(\ref{eq:decompose}) as a summation over a sequence of vectors $\vez_i$, where each vector $\vez_i$ equals $\vece_h$ or $\vemd(j,k)$ or $\overline{\vemd}(j,k)$ or $\veod(h)$. Hence, we can rewrite  Eq~(\ref{eq:decompose}) as 
\begin{eqnarray}\label{eq:decompose-z}
\tilde{\vey}=\sum_i \vez_i.
\end{eqnarray}
Consequently,
\begin{eqnarray}\label{eq:decompose-z-1} 
Rd(\tilde{\vey})=\sum_i Rd(\vez_i).
\end{eqnarray}

We define $Rd(\vex)[\bar{1}]$ as the projection of the vector $Rd(\vex)$ onto the subspace by excluding $Rd^1(\vex)$. Hence, we have
$$Rd(\tilde{\vey})[\bar{1}]=\sum_i Rd(\vez_i)[\bar{1}].$$
According to Observation~\ref{obs:size-2}, we have $\|Rd(\vez_i)[\bar{1}]\|_{\infty}\le \OFPT(1)$, whereas by Lemma~\ref{lemma:merging-lemma} we can find disjoint subsets $T_1,T_2,\cdots,T_{m'}$ such that $|T_j|=\OFPT(1)$ and $\sum_{i\in T_j}Rd(\vez_i)[\bar{1}]\sqsubseteq Rd(\tilde{\vey})[\bar{1}]$ and $Rd(\tilde{\vey})[\bar{1}]=\sum_j(\sum_{i\in T_j}Rd(\vez_i)[\bar{1}])$.

Now we consider $Rd^1(\vex)$'s. By Eq~(\ref{eq:decompose-z-1}) we have 
\begin{eqnarray*}\label{eq:decompose-z-2} 
	Rd^1(\tilde{\vey})=\sum_{j=1}^{m'} \sum_{i\in T_j}Rd^1(\vez_i).
\end{eqnarray*}
Given that $Rd^1(\vemd(j,k)),Rd^1(\overline{\vemd}(j,k))=\OFPT(n)$, $Rd(\vece_h)=\OFPT(1)$, and $|T_j|=\OFPT(1)$, we can conclude that $\|\sum_{i\in T_j}Rd^1(\vez_i)\|_{\infty}=\OFPT(n)$. Applying Lemma~\ref{lemma:merging-lemma}, we can further find $m''$ disjoint sets $T_1',T_2',\cdots,T_{m''}'\subseteq\{1,2,\cdots,m'\}$ such that $|T_h'|=\OFPT(n^{t_A^2})$, $\cup_{h=1}^{m''}=\{1,2,\cdots,m'\}$ and $\sum_{j\in T_h'}\sum_{i\in T_j}Rd^1(\vez_i)\sqsubseteq Rd^1(\tilde{\vey})$. Hence, Eq~(\ref{eq:decompose-z}) can be rewritten as:
$$\tilde{\vey}=\sum_{h=1}^{m''}\left(\sum_{j\in T_h'}\sum_{i\in T_j}\vez_i\right),$$
where for every $h$ it holds that $\sum_{j\in T_h'}\sum_{i\in T_j}\vez_i\sqsubseteq \tilde{\vey}$, $\|\sum_{j\in T_h'}\sum_{i\in T_j}\vez_i\sqsubseteq \tilde{\vey}\|_{\infty}=\OFPT(n^{t_A^2})$, i.e., the following lemma is true.
\begin{lemma}\label{lemma:decompose-tildevey}
	Let $H_0\vey=0$ and $\tilde{\vey}$ be the equalization of $\vey$, then there exist $\vez_h$'s such that $H_0\vez_h=0$, $\vez_h\sqsubseteq \tilde{\vey}$, $\|\vez_h\|_{\infty}=\OFPT(n^{t_A^2})$ and $\tilde{\vey}=\sum_{h=1}^{m''}\vez_h$. Furthermore, the $n+1$ bricks of each $\vez_h$ can be divided into $2+2(\lambda+1)\sigma=\OFPT(1)$ groups such that for any $i_1,i_2\in Gr_\ell$, $\vez_h^{i_1}=\vez_{h}^{i_2}$, and $\vey^{i_1}$, $\vey^{i_2}$ have the same type.
\end{lemma}

\subsubsection{A sign-compatible decomposition of $\vey$}\label{subsec:decompose-vey}
Let $\Gamma=\sum_{k=1}^{\lambda}\|\vev_k\|_{\infty}=\OFPT(1)$. Let $\vez_h$'s be the same as that in Lemma~\ref{lemma:decompose-tildevey}. The goal of this subsection is to prove the following lemma.
\begin{lemma}\label{lemma:decompose-vey}
	If $m''>2\Gamma\cdot  t_A\cdot \left( 2+2(\lambda+1)\sigma\right)$, then there exists some $h_0$ such that $\vez_{h_0}\sqsubseteq \vey$.
\end{lemma}

Towards the proof, we need the following observation and lemma. For an arbitrary $(t_B+nt_A)$-dimensional vector $\vez$, we define by $\vez^i[j]$ the $j$-th coordinate of the brick $\vez^i$. 
Recall the definition of $\vey_f$. As the average is taken among bricks of the same type, we have the following observation.
\begin{observation}
	If $\vey^i[j]$ is large, then $|\vey_f^i[j]|>\Gamma$. Otherwise, $|\vey_f^i[j]|\le \Gamma$.
\end{observation}
By Lemma~\ref{lemma:average}, we have the following corollary.
\begin{corollary}\label{coro:type}
	\begin{itemize}
		\item If $\vey^i[j]$ is positive large, then $\tilde{\vey}^i[j]>0$.
		\item If $\vey^i[j]$ is negative large, then 
		$\tilde{\vey}^i[j]<0$.
		\item If $\vey^i[j]$ is small, then 
		$|\tilde{\vey}^i[j]|\le 2\Gamma$.
	\end{itemize}
\end{corollary}

Using the above corollary, we have the following lemma, which implies directly Lemma~\ref{lemma:decompose-vey}.
\begin{lemma}\label{lemma:aug-1}
	If $m''>2\Gamma t_A\cdot \left( 2+2(\lambda+1)\sigma\right)$, then there exists some $\vez_{h_0}$ such that  
	\begin{itemize}
		\item If $\vey^i[j]$ is positive large, then $\vez_{h_0}^i[j]\ge 0$.
		\item If $\vey^i[j]$ is negative large, then $\vez_{h_0}^i[j]\le 0$.
		\item If $\vey^i[j]$ is small, then $\vez_{h_0}^i[j]= 0$.
	\end{itemize}
\end{lemma}
\begin{proof}
	First, by Lemma~\ref{lemma:decompose-tildevey} we have $\vez_h\sqsubseteq \tilde{\vey}$ for every $1\le h\le m''$. If  $\vey^i[j]$ is positive large, by Corollary~\ref{coro:type} we have $\tilde{\vey}^i[j]>0$, then $\vez_h^i[j]\ge 0$. Similarly if $\vey^i[j]$ is negative large we have $\vez_h^i[j]\le 0$. It remains to consider small coordinates. Consider the following set:
	$$Z_s=\{h: \exists 1\le i\le n, 1\le j\le t_A \text{ such that } \vey^i[j] \text{ is small and } \vez_h^i[j]\neq 0\}.$$
	We claim that, $|Z_s|\le 2\Gamma t_A\cdot \left( 2+2(\lambda+1)\sigma\right)$. Suppose on the contrary that this claim is not true, then $Z_s$ contains more than $2\Gamma t_A\cdot \left( 2+2(\lambda+1)\sigma\right)$ elements, and consequently there exists some $1\le \ell_0\le 2+2(\lambda+1)\sigma$ such that $|Z_s\cap Gr_{\ell_0}|>2\Gamma t_A$. As $1\le j\le t_A$, there exists some $j_0$ such that 
	$$|\{h: \vey^i[j_0] \text{ is small and } \vez_h^i[j_0]\neq 0, i\in Gr_{\ell_0}\}|>2\Gamma.$$
	Note that for all $i\in Gr_{\ell_0}$, $\vez_h^i[j_0]$ takes the same value, hence, for an arbitrary $i_0\in Gr_{\ell_0}$ we have that
	$$|\{h: \vey^{i_0}[j_0] \text{ is small and } \vez_h^{i_0}[j_0]\neq 0\}|>2\Gamma.$$	
	Let $Z_s[j_0]=\{h: \vey^{i_0}[j_0] \text{ is small and } \vez_h^{i_0}[j_0]\neq 0\}$. According to Corollary~\ref{coro:type}, $|\tilde{\vey}^{i_0}[j_0]|\le 2\Gamma$. Meanwhile the fact that $\vez_h\sqsubseteq \tilde{\vey}$ implies that either $\vez_h^{i_0}[j_0]>0$ for all $h\in Z_s[j_0]$, or $\vez_h^{i_0}[j_0]<0$ for all $h\in Z_s[j_0]$. In either case, we conclude that $|\sum_{h\in Z_s[j_0]}\vez_h^{i_0}[j_0]|>2\Gamma$. As $\tilde{\vey}=\sum_h\vez_h$ is a sign-compatible decomposition, we have $|\tilde{\vey}^{i_0}[j_0]|>2\Gamma$, which is a contradiction. Hence, $|Z_s|\le 2\Gamma t_A\cdot \left( 2+2(\lambda+1)\sigma\right)$. Thus, if $m''>2\Gamma t_A\cdot \left( 2+2(\lambda+1)\sigma\right)$, there must exist some $h_0$ such that $\vez_{h_0}^i[j]=0$ for all $i,j$ where $\vey^i[j]$ is small.
\end{proof}
It is clear that the $\vez_{h_0}$ in Lemma~\ref{lemma:aug-1} satisfies that $\vez_{h_0}\sqsubseteq \vey$, whereas Lemma~\ref{lemma:decompose-vey} is proved.

Recall that $\|\vez_h\|_{\infty}=\OFPT(n^{t_A^2})$, then there exists some function $f(A,B,C,D)$ that only depends on the small matrices $A,B,C,D$ (or more precisely, the parameters $\Delta,s_A,s_B,s_C,s_D,t_A,t_B,t_C,t_D$) such that $\|\vez_h\|_{\infty}=f(A,B,C,D)\cdot n^{t_A^2}$. Consequently, the following corollary follows directly from Lemma~\ref{lemma:decompose-vey}.

\begin{corollary}\label{coro:vey-decompose}
	If $\|\vey\|_{1}>2\Gamma\cdot  t_A\cdot \left( 2+2(\lambda+1)\sigma\right)\cdot (t_B+nt_A)\cdot f(A,B,C,D) n^{t_A^2}+2\Gamma\cdot (t_B+nt_A)$, then there exists some $\vez_{h_0}$ such that $\|\vez_{h_0}\|_{\infty}\le f(A,B,C,D) n^{t_A^2}$ and $\vez_{h_0}\sqsubseteq \vey$.
\end{corollary}
\begin{proof}
	Recall Eq~(\ref{eq:average}), we have $n_\ell\cdot \vey_f^k=\sum_{i\in N_\ell} \vey^i$ for all $k\in N_\ell$. Note that $\vey^i$'s have the same type for $i\in N_\ell$, implying for $1\le j\le  t_A$, if $\vey^i[j]$ is large for some $i\in N_\ell$, then $\vey^i[j]$'s are all positive or all negative. This means, for a large coordinate $j$ we have $|\sum_{i\in N_\ell}\vey^i_f[j]|=|\sum_{i\in N_\ell}\vey^i[j]|=\sum_{i\in N_\ell}|\vey^i[j]|$. Hence,
	$$\sum_{i\in N_\ell}\|\vey_f^i\|_1\ge \sum_{i\in N_\ell}\|\vey^i\|_1-\Gamma\cdot n_\ell \cdot t_A$$ 	
	According to Lemma~\ref{lemma:average}, we know that $\|\tilde{\vey}^i-\vey_f^i\|_{\infty}\le \Gamma$, hence 
	$$\sum_{i\in N_\ell}\|\tilde{\vey}^i\|_1\ge \sum_{i\in N_\ell}\|\vey_f^i\|_1-\Gamma\cdot n_\ell \cdot t_A\ge \sum_{i\in N_\ell}\|\vey^i\|_1-2\Gamma\cdot n_\ell \cdot t_A, \quad \forall 1\le i\le n$$ 
	Recall that $\tilde{\vey}^0=\vey^0$, hence,
	$$\|\tilde{\vey}\|_1\ge \|\vey\|_1-2\Gamma\cdot (t_B+nt_A).$$
	
	If $\|\vey\|_{1}>2\Gamma\cdot  t_A\cdot \left( 2+2(\lambda+1)\sigma\right)\cdot (t_B+nt_A)\cdot f(A,B,C,D)n^{t_A^2}+2\Gamma\cdot (t_B+nt_A)$, then 
	\begin{eqnarray}\label{eq:vey-decompose}
	\|\tilde{\vey}\|_1> 2\Gamma\cdot  t_A\cdot \left( 2+2(\lambda+1)\sigma\right)\cdot (t_B+nt_A)\cdot f(A,B,C,D)n^{t_A^2}.	
	\end{eqnarray}
	Since $\|\vez_h\|_{\infty}\le f(A,B,C,D) n^{t_A^2}$, we have $\|\vez_h\|_{1}\le f(A,B,C,D) n^{t_A^2} \cdot (t_B+nt_A)$. As $\tilde{\vey}=\sum_{h=1}^{m''}\vez_h$,  Eq~(\ref{eq:vey-decompose}) implies that $m''> 2\Gamma\cdot  t_A\cdot \left( 2+2(\lambda+1)\sigma\right)$. By Lemma~\ref{lemma:decompose-vey}, there exists some $\vez_{h_0}\sqsubseteq \vey$.	
\end{proof}

By Corollary~\ref{coro:vey-decompose} and the definition of Graver basis, we know that if $\|\vey\|_{1}>2\Gamma\cdot  t_A\cdot \left( 2+2(\lambda+1)\sigma\right)\cdot (t_B+nt_A)\cdot f(A,B,C,D) n^{t_A^2}+2\Gamma\cdot (t_B+nt_A)=\Omega_{FPT}(n^{t_A^2+1})$, then $\vey$ is not a Graver basis element. Hence, Theorem~\ref{thm:3-block-graver} is true.

\section{Proof of Theorem~\ref{thm:4block-3block}}
\begin{reptheorem}{thm:4block-3block}
	Any $4$-block $n$-fold IP with parameters $\Delta, s_A, s_B, s_C, s_D, t_A, t_B, t_C, t_D$ has a feasibly kernel-preserving extended formulation whose constraint matrix is a $3$-block $n$-fold matrix with parameters $\hat{\Delta}, \hat{s}_A, \hat{s}_B, \hat{s}_D, \hat{t}_A, \hat{t}_B, \hat{t}_D$ satisfying 
	\begin{align*}
	\hat{\Delta} &= \Delta & \hat{t}_A &= \hat{t}_D =  2t_C + t_D + s_A  & \hat{t}_B &= t_B & 
	\hat{s}_A = \hat{s}_B = s_B + t_C && \hat{s}_D &= s_D = s_C \enspace .
	\end{align*}
\end{reptheorem}
\begin{proof}
	Let us construct a $3$-block $n$-fold IP instance which models the given $4$-block IP instance.
	It's matrix $\hat{H}_0$ is a $3$-block $n$-fold matrix composed of blocks $\hat{A}, \hat{B}$ and $\hat{D}$, and the remaining data is $\hat{\veb}, \hat{\vel}, \hat{\veu}$ and $\hat{\vew}$.
	Let the blocks be defined as follows.
	\begin{align*}
	\hat{D} &= (C~D~\vezero~\vezero) & \hat{A} &= \left(\begin{smallmatrix}-I~ \vezero ~ I ~ \vezero \\ \vezero ~ A ~ \vezero ~ I\end{smallmatrix} \right) & \hat{B} &= \left(\begin{smallmatrix} I \\ B\end{smallmatrix} \right)
	\end{align*}
	We call the four columns of $\hat{A}$ and $\hat{D}$ \emph{subbricks} and index them by greek letters $\alpha, \beta, \gamma$ and $\delta$, i.e., $\vex^{1\alpha}$ is the $\alpha$-subbrick of the first brick.
	
	Now, we add an extra brick which we call an \emph{aggregation} brick, denoted $\vex^{d}$ where $d=n+1$. 
	The idea is that the $\alpha$ subbrick is non-zero only at the aggregation brick and corresponds to the first-stage variables of the original $4$-block $n$-fold IP.
	We shall ensure that this is true using lower and upper bounds.
	However, to subsequently ``assign'' the aggregated values to the first stage variables, we also need to modify the $B$ block, which, in turn, forces us to introduce new slack variables.
	This is the meaning of the $\gamma$ subbrick (slack variables for bricks $i \neq d$) and $\delta$ subbrick (slack variables for the $a$ brick).
	
	The right hand side $\hat{\veb}$ is simply $\hat{\veb}^0 = \veb^0$ and $\hat{\veb}^i = (\vezero ~ \veb^i)$ for $i \neq d$ and $\hat{\veb}^d = (\vezero ~ \vezero)$.
	We set the new lower and upper bounds $\hat{\vel}, \hat{\veu}$ as follows:
	\begin{description}
		\item[$\alpha$ subbrick] $\hat{\vel}^{i\alpha} = \hat{\veu}^{i\alpha} = \vezero$ for all $i \neq d$, and $\hat{\vel}^{d\alpha} = -\infty$, $\hat{\veu}^{d\alpha} = +\infty$.
		This ensures the $\alpha$ subbrick to be only possibly non-zero in brick $d$.
		\item[$\beta$ subbrick] $\hat{\vel}^{i\beta} = \vel^i$ and $\hat{\veu}^{i\beta} = \veu^i$ for all $i \neq d$ and $\hat{\vel}^{d\beta} = \hat{\veu}^{d\beta} = \vezero$.
		This ensures that the $\beta$ subbrick has the meaning of the original variables $\vex^i$ for all bricks except brick $a$, where we enforce $\hat{\vex}^{d\beta} = \vezero$.
		\item[$\gamma$ subbrick] $\hat{\vel}^{d\gamma} = \hat{\veu}^{d\gamma} = \vezero$ and $\hat{\vel}^{i\gamma} = -\infty$, $\hat{\veu}^{i\gamma} = +\infty$ for $i \neq d$. 
		Without these variables and due to the structure of $\hat{A}$ and $\hat{B}$, we would be enforcing for \emph{each} brick $i \neq d$ that $\hat{\vex}^{0} = \hat{\vex}^{i\alpha}$, and since $\hat{\vex}^{i\alpha} = \vezero$ this would mean $\hat{\vex}^0 = \vezero$. The $\gamma$ subbrick relaxes this to $\hat{\vex}^0 = \hat{\vex}^{i\alpha} + \hat{\vex}^{i\gamma} = \vezero + \hat{\vex}^{i\gamma}$ which is trivially satisfiable considering our setting of the bounds $\hat{\vel}^{d\gamma}$ and $\hat{\veu}^{d\gamma}$.
		\item[$\delta$ subbrick] $\hat{\vel}^{i\delta} = \hat{\veu}^{i\delta} = \vezero$ for all $i \neq a$, and $\hat{\vel}^{d\delta} = -\infty$, $\hat{\veu}^{d\delta} = +\infty$, i.e., the same as for the $\alpha$ subbrick.
		Similarly to the $\gamma$ subbrick, without the $\delta$ subbrick we would be enforcing $B\hat{\vex}^0 + A\hat{\vex}^{d\beta} = \vezero$, however $\hat{\vex}^{d\beta} = \vezero$ so we would be forcing $B\hat{\vex}^0 = \vezero$, which is undesired.
		Thus we relax it to $B\hat{\vex}^0 + A\hat{\vex}^{d\beta} + \hat{\vex}^{d\delta} = B\hat{\vex}^0 + \hat{\vex}^{d\delta} = \vezero$ which is trivially satisfiable.
	\end{description}
	
	To show that the constructed system
	\begin{equation}
	H^0 \hat{\vex} = \hat{\veb}, \, \hat{\vel} \leq \hat{\vex} \leq \hat{\veu}, \, \hat{\vex} \in \Z^{\hat{t}_B + (n+1)\hat{t}_A} \label{eq:3block4block}
	\end{equation}
	is truly an extended formulation of $H\vex = \veb, \, \vel \leq \vex \leq \veu, \, \vex \in \Z^{t_C + nt_A}$, let us define a projection $\pi: \Z^{\hat{t}_B + (n+1)\hat{t}_A} \to \Z^{t_C + n t_A}$ which defines the mapping from the extended formulation to the original instance.
	Specifically, we let
	$$\pi((\hat{x}^0, \hat{x}^{1\alpha}, \hat{x}^{1\beta}, \hat{x}^{1\gamma}, \hat{x}^{1\delta}, \hat{x}^{2\alpha}, \dots, \hat{x}^{n\delta}, \hat{x}^{a\alpha}, \dots, \hat{x}^{a\delta}) = (\hat{x}^0, \hat{x}^{1\beta}, \hat{x}^{2\beta}, \dots, \hat{x}^{n\beta}) \enspace .$$
	By the arguments above we see that $\hat{\vex}^0$ has precisely the meaning of $\vex^0$ and $\hat{\vex}^{i \beta}$ for $i \neq a$ has the meaning of $\vex^i$.
	
	Finally, let us argue that this extended formulation is also feasibly kernel-preserving.
	Consider now a feasible solution $\hat{\vex}$ of~\eqref{eq:3block4block}, and consider any $\hat{\veg}$ in $\ker(H_0)$ such that $\hat{\vex} + \hat{\veg}$ is again feasible.
	We have to show that $H\pi(\hat{\vex}) = \vezero$ and $\vel \leq \vex + \pi(\hat{\vex}) \leq \veu$.
	The latter follows easily from the fact that $\hat{\vel} \leq \hat{\vex} + \hat{\veg} \leq \hat{\veu}$ and that $\pi(\hat{\vel}) = \vel$ and $\pi(\hat{\veu}) = \veu$.
	To see the former, consider separately first the upper row $(C~D~\cdots~D)$ of $H$ and after that the remaining rows.
	We have that
	\begin{equation*}
	C \hat{\vex}^{a\alpha} + D \hat{\vex}^{a \beta} + \vezero \hat{\vex}^{a \gamma} + \vezero \hat{\vex}^{a \delta} + \sum_{i=1}^n C \hat{\vex}^{i\alpha} + D \hat{\vex}^{i\beta} + \vezero \hat{\vex}^{i\gamma} + \vezero \hat{\vex}^{i\delta} = \vezero \enspace .
	\end{equation*}
	Omitting the zero blocks, we obtain
	\begin{equation*}
	C \hat{\vex}^{a\alpha} + D \hat{\vex}^{a \beta} + \sum_{i=1}^n C \hat{\vex}^{i\alpha} + D \hat{\vex}^{i\beta} = \vezero \enspace .
	\end{equation*}
	Recall that our bounds enforce $\hat{\vex}^{a \beta} = \vezero$ and $\hat{\vex}^{i \alpha} = \vezero$ for $i \neq a$, and finally $\hat{\vex}^0 = \hat{\vex}^{a \alpha}$, so plugging these in we obtain
	\begin{equation*}
	C \hat{\vex}^{0} + \sum_{i=1}^n D \hat{\vex}^{i\beta} = \vezero,
	\end{equation*}
	which by the definition of $\pi$ implies that $C \pi(\vex)^0 + \sum_{i=1}^n D \pi(\vex)^i = \vezero$ as desired.
	Now it is left to show that, for each $i \neq a$, $B \pi(\vex)^0 + A \pi(\vex)^i = \vezero$.
	We have that
	\begin{equation*}
	B \hat{\vex}^0 + \vezero \hat{\vex}^{i\alpha} + A \hat{\vex}^{i\beta} + \vezero \hat{\vex}^{i \gamma} + I \hat{\vex}^{i \delta} = \vezero \enspace .
	\end{equation*}
	Omitting the zero blocks and recalling that our bounds enforce $\hat{\vex}^{i \delta} = \vezero$ for each $i \neq a$, we have
	\begin{equation*}
	A \hat{\vex}^0 + A \hat{\vex}^{i \beta} = \vezero,
	\end{equation*}
	which, by definition of $\pi$, is what we wanted to show.
\end{proof}

\section{Proof of Theorem~\ref{thm:3-block-better-lower}}
\begin{reptheorem}{thm:3-block-better-lower}
	There exists an instance of 3-block $n$-fold IP with a matrix $H_0$ 
	where $\hat{s}_A=\hat{s}_B=2t, \hat{s}_D=t-1,\hat{t}_A=\hat{t}_D=4t$, $\hat{t}_B=t$ such that for every feasible solution $\vex$ and for every $\veg \in \ker_{\Z}(H_0)$ which is feasible with respect to $\vex$, it holds that $\|\veg\|_\infty = \Omega(n^{t-1})$, and in particular, $\|\veg^0\|_\infty = \Omega(n^{t-1})$. 
\end{reptheorem}
	\begin{proof}
		Consider the instance constructed in Theorem~\ref{thm:4-block-lower} with $H$ being the 4-block $n$-fold matrix from the proof. 
		Apply Theorem~\ref{thm:4block-3block} to this instance to obtain its feasible kernel-preserving extended formulation, which is a 3-block $n$-fold IP, and consider any $\hat{\vex}$ which is a feasible solution for it.
		Denote by $\pi$ the projection from the proof of Theorem~\ref{thm:4block-3block}.
		
		Now let $\hat{\veg} \in \ker_{\Z}(H_0) \subseteq \ker(H_0)$ be feasible with respect to $\hat{\vex}$.
		By Definition~\ref{def:fkp-ef}, we have $\veg = \pi(\hat{\veg}) \in \ker_{\Z}(H)$, and by Theorem~\ref{thm:4-block-lower} we have $\|\veg\|_\infty = \Omega(n^{t-1})$ and in particular $\|\veg^0\|_\infty = \Omega(n^{t-1})$.
		By the definition of $\pi$ these lower bounds transfer to $\hat{\veg}$ and $\hat{\veg}^0$.
	\end{proof}

\section{Proof of Theorem~\ref{thm:alg-3-block} and Theorem~\ref{thm:alg-4-block}}
\begin{reptheorem}{thm:alg-3-block}
	There exists an algorithm for 3-block $n$-fold IP that runs in $\min\{\OFPT(n^{s_{\scalebox{.5}{D}}t_B+3}\log^3 n), \OFPT(n^{(t_A^2+1)t_B+3}\log^3 n)\}$ time.
\end{reptheorem}
\begin{proof}
Using the idea of approximate Graver-best oracle introduced by Altmanová et al.~\cite{altmanova2018evaluating} and implicitly by Eisenbrand et al.~\cite{eisenbrand2018faster}, it suffices for us to solve the following IP for each fixed value $\rho_0=2^0,2^1,2^2,\cdots$:
	\begin{eqnarray*}\label{eq:graver-best-aug-3-block}
		\min\{\vew\cdot \vex: H_0 \vex=0, \vel\le \vex_0+\rho_0\vex\le \veu, \vex\in\Z^m, \|\vex\|_{\infty}\le \min\{\OFPT(n^{s_c}),\OFPT(n^{t_A^2+1})\}\}
	\end{eqnarray*}
	Let $\vex_*$ be the optimal solution. Given that $\|\vex_*\|_{\infty}\le \OFPT(n^{t_A^2+1})$, we can guess $\vex^0_*$ and there are $\OFPT(n^{(t_A^2+1)t_B})$ different possibilities. For each guess, say, $\vex_*^0=\vev$, we solve the following problem:
	\begin{eqnarray*}\label{eq:graver-best-aug-3-block-1}
		\min\{\vew\cdot \vex: H_0 \vex=0, \vel\le \vex_0+\rho_0\vex\le \veu, \vex\in\Z^m, \vex^0=\vev\}
	\end{eqnarray*}
	By fixing $\vex^0$, the above problem becomes exactly an $n$-fold IP, which can be solved efficiently in $\OFPT(n^2\log n^2)$ time~\cite{eisenbrand2018faster}. Notice that $\rho_0$ may take $\OFPT(n\log n)$ distinct values, the overall running time is $\min\{\OFPT(n^{s_{\scalebox{.5}{D}}t_B+3})\log^3 n, \OFPT(n^{(t_A^2+1)t_B+3}\log^3 n)\}$.
\end{proof}
Theorem~\ref{thm:alg-4-block} can be proved by plugging in the upper bound of 4-block $n$-fold IP and proceed with the same argument.

\bibliographystyle{plain}
\bibliography{schedule-tree}

\end{document}